\documentclass[12pt]{article}

\usepackage[a4paper, margin=1.25 in]{geometry}
\usepackage[latin1]{inputenc}
\usepackage{hyperref}
\usepackage{graphics}
\usepackage{threeparttable}
\usepackage{epsfig}
\usepackage{algorithm,algpseudocode}
\usepackage{amscd,amsfonts,amsopn,amssymb,amstext,amsmath}
\usepackage{appendix}
\usepackage{booktabs}
\usepackage{amssymb}
\usepackage{epstopdf}
\usepackage{natbib}
\usepackage{fullpage}
\usepackage{graphicx,graphics,psfrag}
\usepackage{latexsym,enumerate}
\usepackage{multirow}
\usepackage{natbib}
\usepackage{subcaption}
\usepackage{rotating}
\usepackage{setspace}
\usepackage[T1]{fontenc}
\usepackage{times}
\usepackage{url}
\usepackage{bm}
\usepackage{soul}

\begin{document}

\newtheorem{definition}{Definition}
\newtheorem{theorem}{Theorem}
\newtheorem{example}{Example}
\newtheorem{corollary}{Corollary}
\newtheorem{lemma}{Lemma}
\newtheorem{proposition}{Proposition}
\newtheorem{remark}{Remark}

\newenvironment{proof}{{\bf Proof:\ \ }}{\qed}
\newcommand{\qed}{\rule{0.5em}{1.5ex}}
\newcommand{\bfg}[1]{\mbox{\boldmath $#1$\unboldmath}}

\numberwithin{proposition}{section}
\numberwithin{equation}{section}
\numberwithin{theorem}{section}
\numberwithin{example}{section}
\numberwithin{definition}{section}
\numberwithin{lemma}{section}
\numberwithin{remark}{section}

\newcommand{\R}{\mathbb{R}}
\newcommand{\argmin}{\arg\!\min}
\newcommand{\diff}{\mathop{}\!\mathrm{d}}

\setstretch{1.5}   

\begin{center}

\section*{Beyond the Beta Lorenz Curve: A New Parametric Family for Poverty and Inequality Estimation}

\vskip 0.2in {\small \bf Jos\'e Mar\'{\i}a Sarabia$\,^a$\footnote{Corresponding
author. E-mail address: sarabiaj@unican.es (JM Sarabia), vanesa.jorda@unican.es (V. Jord\'a), emilio.gomez-deniz@ulpgc.es (E. G\'omez-D\'eniz).}, Vanesa Jord\'a$\,^a$, Emilio G\'omez-D\'eniz$\,^b$
\vskip 0.2in
{\small\it $\,^a$Department of Economics and SANFI,  University of Cantabria, Spain}\\
{\small\it $\,^b$Department of Quantitative Methods in Economics and TIDES Institute,\\
University of Las Palmas de Gran Canaria, Spain}\\[-0.2cm]
}

\end{center}

\begin{abstract}\noindent

The estimation of inequality and poverty measures is frequently constrained by a lack of individual data. When only income shares are available, the Beta Lorenz curve introduced by Kakwani (\textit{Econometrica}, 48, 1980) has become a standard tool for reconstructing income distributions. Together with the General Quadratic (GQ) Lorenz curve, it is used by the World Bank to produce official poverty  estimates when microdata are unavailable. In this paper, we show that Kakwani's specification does not generally satisfy the formal  requirements of a genuine Lorenz curve and introduce a new Lorenz curve rooted in the corrected parameter space. Using more than 1,700 datasets, we show that the proposed model yields valid Lorenz curves in all applications and consistently outperforms competing specifications in both accuracy and sampling precision. The GQ, by contrast, fails to provide genuine curves in 15 percent of datasets and underestimates extreme poverty severity in 98 percent of them.

\end{abstract}

\noindent {\bf Key Words}: Grouped data, parametric models, Gini index, poverty gap, Watts index

\section{Introduction}
The eradication of poverty and the reduction of inequality sit at the very heart of the 2030 Agenda for Sustainable Development and its global milestones. To monitor progress toward these goals, it is essential to measure poverty and inequality on a regular basis using high-quality data. To this end, several international initiatives have emerged to expand data availability and improve methodological harmonization across countries and over time. This is the role played by the Poverty and Inequality Platform (PIP) developed by the World Bank, which offers regularly updated indicators for cross-country comparisons and intertemporal analyses of the evolution of poverty and inequality. Complementary to this initiative, the World Inequality Database (WID) offers detailed distributional estimates through the systematic integration of available data from surveys, tax records, and national accounts.

These efforts, however, are often undermined by the limited availability of individual-level data. Even today, no microdata on individual incomes are available for several countries, including Algeria, China, Guyana, Turkmenistan, Trinidad and Tobago, Venezuela, and the United Arab Emirates, and data limitations are even more pronounced for earlier periods \citep{wb_handbook}.\footnote{In the case of China, although survey projects such as the China Household Income Project (CHIP) collect detailed household-level information, access to microdata is restricted and not systematically available for international research or long historical comparisons. Consequently, for most empirical applications, researchers must rely on grouped income shares (e.g., quintiles or deciles), which constitute the only consistently accessible source of distributional information.} When only grouped or tabulated data are available, parametric models are a valuable tool for reconstructing the underlying income distribution, from which poverty and inequality measures can be estimated \citep{jorda2021}.

In practice, distributional estimates in the WID are frequently constructed using the methods proposed by \citet{blanchet2022}, which rely on flexible models from the generalized Pareto family. Generalized Pareto models are an effective tool for characterizing distributions that exhibit a power-law tail \citep{arnold2008}. Consequently, they are an attractive candidate for modeling income distribution above the median. However, their application may be less effective for the lower half of the distribution. As a result, using this model to estimate poverty and inequality would likely lead to significant biases.

By contrast, the World Bank's official poverty estimates rely on parametric Lorenz curve specifications when individual-level data are unavailable. The estimation of their parameters is often simpler compared to more traditional statistical distributions, especially because these models only require the estimation of shape parameters. While several specifications exist,\footnote{Numerous models have been proposed to approximate the Lorenz curve \citep[see, e.g.,][]{kakwani1976, rasche1980, villasenor1989, basmann1990, ortega1991, chotikapanich1993, holm1993, arnold1987, ryu1996, sarabia1999, rohde2009}. For comprehensive reviews on the functional forms and properties of parametric Lorenz curves, we refer the reader to \citet{chotikapanich2008} and \citet{arnold2018}.} PIP relies on two primary models: the General Quadratic (GQ) Lorenz curve developed by \citet{villasenor1989} and the Beta Lorenz curve introduced by \citet{kakwani1980}. Although these models have been found to be particularly suitable for estimating poverty indicators \citep{minoiu2009, bresson2009}, they are prone to reporting inconsistent estimates. These inconsistencies are predominantly caused by the restrictive parameter constraints that characterize these models as valid Lorenz curves. When these conditions are violated, the model can produce non-increasing concave curves or even negative values, which do not meet the basic regularity conditions that a theoretically valid Lorenz curve should satisfy.\footnote{According to \citet{bresson2009}, these models do not meet the regularity conditions for a valid Lorenz curve in about half of the empirical income distributions analyzed.}

The Beta Lorenz curve has been particularly problematic in this regard. As we demonstrate in this paper, Kakwani's original proposal does not, in general, satisfy the defining properties of a genuine Lorenz curve and is theoretically valid only for a single configuration of its parameter space. This lack of theoretical robustness is particularly important given the institutional weight of this model within the World Bank's poverty estimates. This model also remains a popular choice for approximating income distributions in recent studies (see \citet{bresson2009, kobayashi2022, ravallion2022, wilson2022}). Consequently, a substantial body of empirical work relies on a functional form that may violate the mathematical requirements of a Lorenz curve. 

In this paper, we identify the unique case within Kakwani's family that remains theoretically consistent 
across the entire unit interval. Building on this corrected framework, we propose a new specification 
whose parameter space is fully characterized by simple boundary 
constraints that guarantee theoretical validity across a wide range of empirical 
income distributions. We evaluate the performance of our proposed specification through an 
extensive empirical assessment using more than 1,700 datasets. 
Our comparative analysis includes three 
alternative specifications rooted in Kakwani's parametric framework 
as well as the GQ Lorenz curve, given its institutional relevance.

Our results indicate that the proposed model satisfies 
the regularity conditions of a genuine Lorenz curve in every dataset 
analyzed, while the GQ model fails in approximately 15 percent of 
empirical applications. Moreover, this new Lorenz curve
consistently delivers the most accurate estimates 
of poverty and inequality measures across all specifications and 
datasets considered, with estimation errors that are not only lower on 
average but also exhibit substantially less variability across datasets 
than those of competing models. Most importantly, our findings reveal 
that the GQ Lorenz curve systematically underestimates the depth and severity of extreme 
poverty. Conversely, 
when higher poverty thresholds are considered, the GQ model tends to 
overestimate poverty incidence in 83 percent of the datasets analized. Finally, Monte Carlo simulations confirm that the 
superiority of $L_3$ in point-estimate accuracy is not achieved at the 
cost of greater sampling uncertainty: $L_3$ exhibits lower standard errors than all specifications 
considered for inequality measures, and sampling uncertainty 
broadly similar to that of GQ for poverty measures.

The remainder of this paper is organized as follows. Section 2 demonstrates the limitations of the Beta Lorenz curve in representing income distributions and introduces alternative models derived from specific cases of this functional form. Section 3 details the estimation methodology for these models. Section 4 evaluates the performance of the proposed specification 
against alternative parametric Lorenz curves. Finally, the paper concludes by discussing the main implications of this study and offering recommendations for the estimation of income and poverty measures from grouped data.

\section{Kakwani's Lorenz curve (1980) revised}\label{section_LC}
The Lorenz curve graphs the cumulative income share $L(p)$ as a function of the cumulative population share $p$ when the income units are arranged according to income size. A characterization of the Lorenz curve attributed to Gaffney and Anstis by \citet{pakes1981} is given by the following theorem.

\begin{theorem}\label{theorem1}
Suppose $L(p)$ is defined and continuous on $[0,1]$ with second derivative $L''(p)$. The function $L(p)$ is a Lorenz curve if and only if,
\begin{equation}\label{LCgenuine}
L(0) = 0,\;\;L(1) = 1,\;\;L'(0^+) \geq 0,\;\;L''(p) \geq 0, \quad p \in (0,1).
\end{equation}
\end{theorem}

Consider the functional form to approximate the Lorenz curve proposed by Kakwani (1980),
\begin{equation}\label{kakwaniLC}
L(p; a, \alpha, \beta) = p - ap^\alpha(1-p)^\beta,
\end{equation}
where
$$
a\ge 0,\;\;0<\alpha\le 1,\;\;0<\beta\le 1.
$$

The conditions required to ensure that (\ref{kakwaniLC}) constitutes a genuine Lorenz curve have been revisited by several authors. For instance, \citet{cheong2002empirical} suggests that Kakwani's specification is a valid Lorenz curve provided that $0 < \alpha, \beta < 1$. However, this result is incorrect. To provide a counterexample, if $a = 1$ and $\alpha = \beta = \frac{1}{2}$, then $L(p; a, \alpha, \beta) \leq 0$ for all $p \in [0, \frac{1}{2}]$, violating the fundamental requirement that a Lorenz curve must lie within the unit square and above the horizontal axis.

The left panel of Figure \ref{lc_example} presents the Beta Lorenz curve for this parameter setting. We observe that a large segment of the curve falls below the horizontal axis. Given that observed income shares are strictly non-negative by definition, one may argue that this is a purely theoretical or extreme scenario that would hardly occur in practice. However, negative Lorenz curves are a common phenomenon when this model is fitted to real-world data . Consider the case illustrated in the right panel of Figure \ref{lc_example}, which depicts the Beta Lorenz curve estimated for Belgian data in 2016. While looking at the whole curve may not reveal any irregularities, a closer inspection of the lower tail reveals a clear violation of the non-negativity constraint. The curve crosses the zero at $p=0.0126$ and remains negative for all values below this threshold.

\begin{figure}[tbhp]
\begin{center}
\includegraphics[scale=0.53]{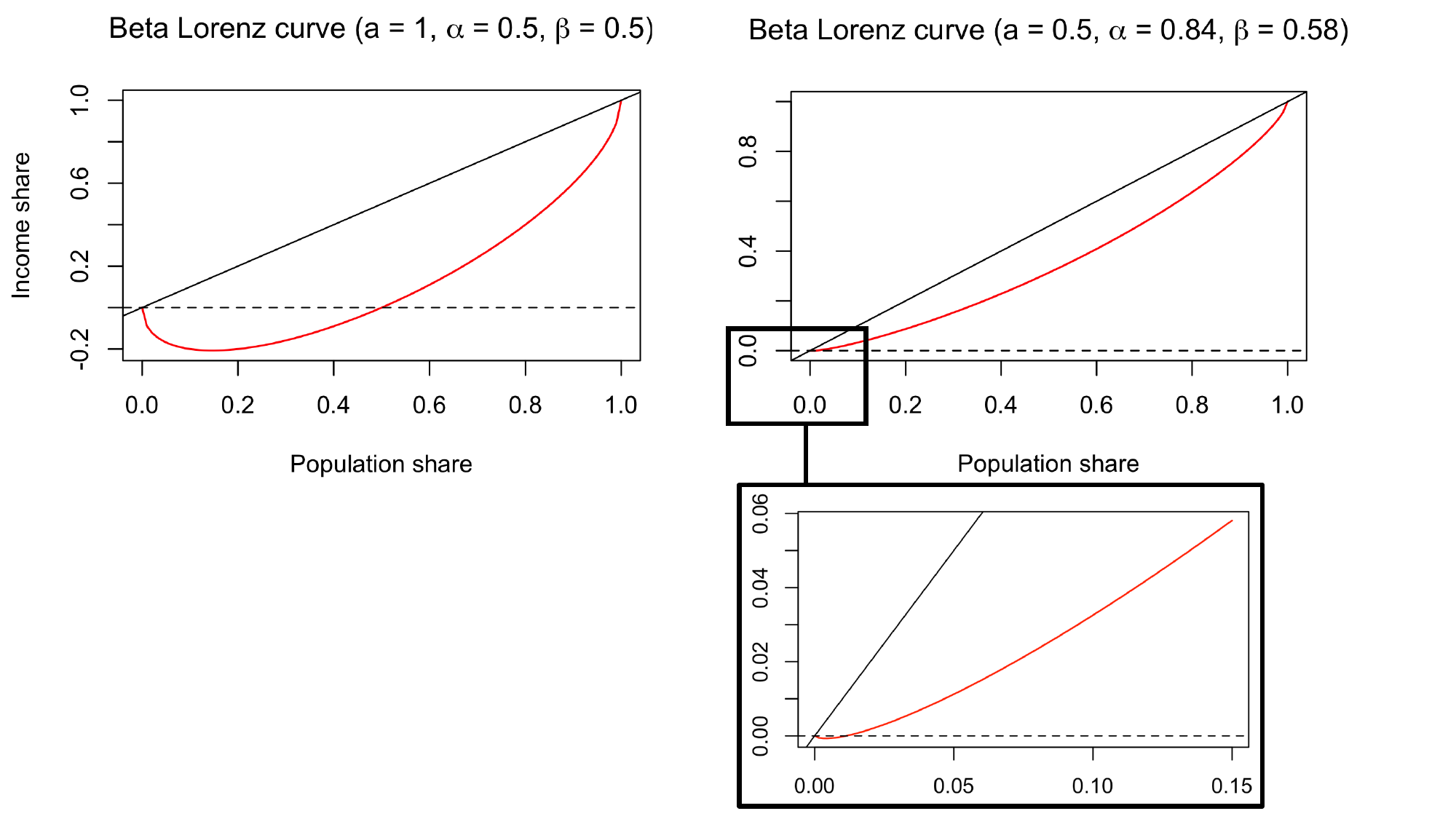}
\caption{Failure of the Beta Lorenz curve to satisfy theoretical requirements} \label{lc_example}
\end{center}
\end{figure}

The literature surrounding the Kakwani proposal contains conflicting results regarding the parameter space required for this model to characterize a mathematically valid Lorenz curve (see, e.g., \citet{schader1994fitting, cheong2002empirical, sarabia1999, sitthiyot2021simple, shen2024regression}). We resolve these inconsistencies by introducing a general theorem that establishes the necessary and sufficient conditions on the parameters $a, \alpha$ and $\beta$ for Eq. \eqref{kakwaniLC} to satisfy the fundamental properties of a Lorenz curve.

\begin{theorem}\label{theorem2}
Consider the Kakwani's Lorenz curve defined (\ref{kakwaniLC}). Then, (\ref{kakwaniLC}) is a genuine Lorenz curve only if
$\alpha = 1$, $0 \leq a\leq 1$ and $0<\beta \leq 1$, which leads to the following curve,
\begin{equation}\label{SpecialCase1}
L_0(p)=p-ap(1-p)^\beta,\;\;0\le p\le 1.
\end{equation}
\end{theorem}
The proof is provided in Appendix A. In the boundary cases, (\ref{SpecialCase1}) is a genuine Lorenz curve if $\beta=0$ and $a=0$; if $\beta=1$ and $0\le a\le 1$; or if $a=1$ and $0<\beta\le 1$.

The Gini index of $L_0(p)$ is given by,
\begin{eqnarray}\label{Gini_L0}
G(a, \beta) &=& 2 \int_0^1 \left[p - L_1(p)\right] dp\nonumber\\
&=& 2aB(2, \beta + 1)\nonumber\\
&=& \frac{2a}{(\beta + 1)(\beta + 2)},\label{GiniSpecialCase1}
\end{eqnarray}
where $B(a,b)$ denotes the beta function. The generalized Gini index and the cumulative distribution function (CDF) associated with  \ref{SpecialCase1} are included in Appendix A.

\section{Functional variants of the Beta Lorenz curve}\label{section2}
\subsection{Alternative Lorenz curve specifications based on Kakwani's parametric framework}\label{section2.1}
In this section, we consider alternative LCs that can be derived from the original Kakwani's proposal. Starting from the original Beta curve in \eqref{kakwaniLC}, we first set the parameters $a=1$ and $\alpha=1$. By factoring out the term $p$ and subsequently generalizing the model by introducing an exponent $a$ (which is conceptually distinct from coefficient $a$ in the original Kakwani specification), we arrive at the \citet{ortega1991} Lorenz curve:

\begin{equation}\label{lco}
L_1(p;a,b)=p^a\left[1-(1-p)^b\right],\;\;0\le p\le 1,
\end{equation}
with,
$$
a\geq 0,\;0<b\le 1.
$$

The Gini index of the curve given in (\ref{lco}) is,
\begin{equation}\label{Gini_L1}
G_{L_1}(a,b)=1-2 \left(\frac{1}{a+1}-\frac{\Gamma (a+1) \Gamma (b+1)}{\Gamma (a+b+2)}\right).
\end{equation}

An extension of (\ref{lco}) is given by the following curve introduced by \citet{sarabia1999}:
\begin{equation}\label{lc2}
L_2(p;a,b,d)=p^a\left[1-(1-p^d)^b\right],\;\;0\le p\le 1,
\end{equation}
with,
$$
a\geq 0,\;0<b\le 1,\;d\ge 1,
$$
whose associated Gini index is given by:
\begin{equation}\label{Gini_L2}
G_{L_2}(a,b,d)=1-2 \left(\frac{1}{a+1}-\frac{\Gamma (b+1) \Gamma \left(\frac{a+1}{d}\right)}{d \Gamma \left(\frac{a+b d+d+1}{d}\right)}\right).
\end{equation}

Since our goal is to obtain genuine and flexible Lorenz curves that yield accurate 
estimates across the entire income distribution, we propose a novel model 
inspired by the properties of upper truncated random variables. Let $X$ be a random 
income variable with CDF $F_X$ and Lorenz curve $L_X$. We define $X_s = \{X \mid 
X \leq s\}$ as the random variable resulting from an upper truncation at income 
level $s$. The Lorenz curve associated with $X_s$ is given by:
\begin{equation}\label{truncatedLC}
L^{(s)}(p) = \frac{L_X(F_X(s)p)}{L_X(F_X(s))}, \quad 0 \leq p \leq 1.
\end{equation}

The introduction of the truncation parameter $s$ has a natural economic 
interpretation that is particularly relevant for poverty measurement. Standard 
Lorenz curve specifications must fit the entire income distribution with a single 
set of shape parameters, creating a tension between accurately representing the 
upper tail, where income concentration is highest, and capturing the lower portion 
of the distribution where poverty is measured. In highly unequal societies, this 
tension is especially acute: the shape parameters are pulled toward accommodating 
the heavy upper tail at the cost of misrepresenting the bottom of the distribution. 
The truncation parameter $s \in (0, 1]$ resolves this tension by isolating the 
influence of the upper tail, freeing the remaining parameters to track the curvature 
of the lower and middle income groups with greater precision. Hence, $s$ 
can be interpreted as a relative income threshold above which the distribution is 
allowed to behave differently from the bulk of the population. 

If $F_X$ is replaced by a valid Lorenz curve, the structure of 
\eqref{truncatedLC} ensures that the resulting function is again a genuine Lorenz 
curve. Our proposed model is derived from \eqref{truncatedLC} by assuming a 
power-function CDF,  $F_X = x^d$ and Ortega's proposal as a base Lorenz curve. 
This choice for $F_X$ is adopted primarily for analytical 
tractability rather than as a substantive claim about the functional form 
of the underlying income distribution. The resulting Lorenz curve is given by:
\begin{equation}\label{lc3b}
L_3(p; a, b, d, s) = \frac{p^{ad} \left[ 1 - (1 - p s^d)^b \right]}
{1 - (1 - s^d)^b}, \quad 0 \leq p \leq 1,
\end{equation}
under the parameter constraints: $$ a \geq 0, \quad 0 < b \leq 1, \quad d \geq 1, \quad 0 < s \leq 1. $$

While $L_3$ is characterized by four parameters, we can only identify 
three independent combinations: 
$\alpha = ad$, $b$, and $\sigma = s^d$. This follows from the fact that  $a$ and $d$ enter $L_3$ 
exclusively through their product and $s$ and $d$ through $s^d$. 
Eq. \eqref{lc3b} admits an equivalent parametrization in terms 
of the identifiable combinations as follows:
\begin{equation}\label{lc3}
L_3(p; \alpha, b, \sigma) = \frac{p^{\alpha}
\left[1-(1-\sigma p)^b\right]}{1-(1-\sigma)^b},
\end{equation}
with $$\alpha > 0  \quad 0 < b \leq 1,  \quad 0 < \sigma \leq 1. $$

While $L_1$ is nested within $L_2$ and $L_3$, it is important to observe that $L_2$ is 
nested within $L_3$ only in the limiting case where $\sigma = 1$ and $d = 1$, 
illustrating that $L_3$ is a distinct generalized parametric model rooted in  
Kakwani's framework. 
 
The Gini index of (\ref{lc3}) is expressed as follows:
\begin{equation}\label{Gini_L3}
G_{L_3}(\alpha, b, \sigma)=1-\frac{2 \left(\, _2F_1\left(-b,\alpha+1;\alpha+2;\sigma\right)-1\right)}{(\alpha+1) \left(\left(1-\sigma\right)^b-1\right)},
\end{equation}
where $\,_2F_1(a,b,c;z)$ is the hypergeometric function.

\subsection{Estimation methods}\label{methods}
In this section, we define the estimation strategy for the models presented in Section \ref{section2.1} from grouped data in the from of income shares. Let \textbf{x} be an $i.i.d.$ random sample of size $N$ from a continuous income distribution $f(x;\boldsymbol\lambda)$ defined over the support $H = [0, \infty)$, with $\boldsymbol\lambda' = (\boldsymbol\theta', \boldsymbol\beta')' \in \boldsymbol\Lambda \subseteq \R ^{k+l},$ where $\boldsymbol\Lambda$ is the parameter space, $\boldsymbol\theta$ is the vector of shape parameters and $\boldsymbol\beta$ is the vector of scale parameters. Assume that $H$ is divided into $J$ mutually exclusive intervals $H_j = [h_{j-1}, h_j) , j= 1, \dots, J$. Denote  $c_j= \sum_{i=1}^N\bm 1_{[h_{j-1}, h_j)}(x_i)x_i/\sum_{i=1}^Nx_i, j=1,\dots, J$ as the proportion of total income held by individuals in the $j^{th}$ interval and the cumulative proportion by $s_j= \sum_{m=1}^j c_m$. Let $p_j = \sum_{i=1}^N\bm 1_{[h_{j-1}, h_j)}(x_i)/N , j=1,\dots, J$ denote the frequency of the sample \textbf{x} in the $j^{th}$ interval and $u_j= \sum_{m=1}^j p_m$ the cumulative frequency. According to this scheme, income shares ($s_j, j=1, \dots, J$) are ordinates of the Lorenz curve corresponding to the abscissae $u_j, j=1, \dots, J$.

In the most comprehensive secondary databases reporting distributional statistics, the sampling design typically dictates that the proportion of observations within each group is predetermined. Consequently, the population shares ($p_j$) are treated as fixed constants, whereas the income shares ($s_j$) are the realized random variables. Under this specific data structure, the maximum likelihood estimation (MLE) framework based on a multinomial specification, as proposed by \citep{mcdonald1984}, is misspecified, as group frequencies are non-stochastic by construction. Furthermore, implementing this MLE approach requires information on the income group boundaries ($h_j$), which are rarely included from secondary sources. Given these data constraints, researchers have traditionally relied on minimum distance estimators to identify the vector of parameters of interest.

These estimators minimize the distance between the observed income shares and the theoretical functional form of the Lorenz curve. Let $X$ be a random variable defined on the support $\mathcal{Z}$, with a cumulative distribution function $F(x; \boldsymbol{\lambda})$ where $\boldsymbol{\lambda} \in \boldsymbol{\Lambda}$. The corresponding Lorenz curve, denoted as $L(u; \boldsymbol{\theta})$, is a function solely of the shape parameters,  $\boldsymbol{\theta} \subseteq \boldsymbol{\lambda}$, which can be estimated as follows:

\begin{equation}\label{nls1}
\hat{\boldsymbol\theta}= \argmin_{\boldsymbol\theta} \bm{M}(\boldsymbol\theta)'\bm{M}(\boldsymbol\theta).
\end{equation}%
$\bm{M}(\boldsymbol\theta)'= [m_1(\boldsymbol\theta), \dots, m_{J-1}(\boldsymbol\theta)]$ is the vector of moment conditions, which takes the form
$$
\bm M(\boldsymbol\theta) = L(\bm u; \boldsymbol\theta) - \bm s,%
$$
where $\bm s' = (s_1, \dots, s_{J-1})$ is a vector of cumulative income shares associated with the population proportions $\bm u' =(u_1, \dots, u_{J-1})$.

This approach is referred to as the Equally Weighted Minimum Distance (EWMD) estimator. Because the Lorenz curve is scale-invariant, the specification in (\ref{nls1}) only permits the identification of the subset of $\boldsymbol{\lambda}$ corresponding to the shape parameters $\boldsymbol{\theta}$. Notably, scale parameters are not needed for the estimation of relative inequality measures, such as the Gini coefficient. Consequently, when the research objective is centered on relative inequality, this estimation strategy circumvents the requirement for mean income data. Conversely, if the objective involves the estimation of poverty measures or absolute welfare, the scale of the distribution can be recovered by incorporating an external estimate of the mean income.

Although EWMD estimators yield consistent estimates of $\boldsymbol{\theta}$, they are generally asymptotically inefficient. This estimator fails to account for two critical features of the data structure. First, the elements of the moment vector $\mathbf{M}(\boldsymbol{\theta})$ may exhibit heteroskedasticity. Second, these moments are inherently correlated because the sum of income shares is, by definition, constrained to one. Consequently, any welfare indicators or indices derived from these parameters, including poverty and inequality measures, will also inherit this loss of efficiency, resulting in wider confidence intervals and reduced statistical power compared to an optimal weighting scheme.

Since our main objective is to provide point estimates for poverty and inequality measures, the EWMD approach remains a tractable and consistent choice. However, if the research objective extends to statistical inference, it should be noted that deriving valid confidence intervals requires data on the sample size, which is rarely available in secondary databases. For those exercises specifically interested in constructing reliable confidence intervals and minimizing the variance of the estimates, the Optimal Minimum Distance estimator is required. See \cite{Hajargasht2016} and \cite{jorda2021} for details. \footnote{For asymptotic efficiency, the inverse of the variance--covariance matrix of the moment conditions should be incorporated into Eq.~(\ref{nls1}) as the optimal weighting matrix. However, when Lorenz curves lack closed-form expressions for their underlying CDFs, implementing this optimal estimator becomes computationally complex and numerically unstable.} 

To implement the EWMD estimator, we use a constrained non-linear least squares 
framework via the PORT optimization algorithm. Although the 
theoretical regularity conditions for a genuine Lorenz curve impose upper bound 
restrictions on certain parameters, we enforce only lower bound constraints during 
optimization, restricting all parameters to be non-negative. The regularity 
conditions of the fitted curve are then evaluated ex post: a model is classified 
as yielding a genuine Lorenz curve if and only if the estimated parameters satisfy 
the conditions established in Sections \ref{section_LC} and \ref{section2.1}. This 
approach minimizes the risk of converging to corner solutions that are technically 
admissible but yield substantially larger estimation errors for poverty and 
inequality measures. By evaluating validity after estimation rather than imposing 
it as an optimization constraint, we ensure that reported regularity violations 
reflect genuine failures of the functional form to represent the underlying 
distribution, rather than artifacts of the optimization procedure.


\section{Evaluating parametric Lorenz curves for poverty and inequality measurement}\label{results}
In this section, we evaluate the empirical performance of alternative 
Lorenz curve specifications in estimating poverty and inequality measures 
from grouped income shares. Specifically, our comparative analysis includes 
the Kakwani special case (\ref{SpecialCase1}), Ortega's specification 
(\ref{lco}), the Sarabia-Castillo-Slottje curve 
(\ref{lc2}), hereafter the SCS curve, and our proposed $L_3$ curve 
(\ref{lc3}), all of which are rooted in Kakwani's parametric framework. 
As a benchmark, we compare these specifications with the GQ Lorenz curve, 
which serves as the primary parametric model used by the World Bank's PIP
 to estimate welfare indicators when microdata are unavailable.

When evaluating model performance, a common approach is to compare the residual sum of squares (RSS). However, this criterion is not appropriate for comparing all specifications considered here because they differ in the number of parameters.\footnote{Nevertheless, RSS remains informative when comparing models with the same number of parameters. Table~\ref{rss} in Appendix B reports the distribution of the RSS for the two-parameter models (Ortega and Kakwani) and the three-parameter specifications (SCS, GQ, and $L_3$). These statistics are computed across the 1,735 datasets from PIP and LIS for which all five Lorenz curve specifications yield genuine estimates. Among the three-parameter models, the proposed $L_3$ specification consistently provides the best fit to the observed income shares. It achieves not only the lowest average RSS but also the smallest value at every reported quantile of the distribution.} More importantly, a model may reproduce the observed Lorenz curve ordinates very closely while still failing to accurately recover other distributional characteristics. Consequently, we evaluate model performance by comparing parametric  estimates of poverty and inequality measures against their empirical counterparts. Although our objective is to assess the performance of these models in settings where only grouped income shares are available, we conduct the evaluation using datasets containing full income microdata. These data allow us to obtain the empirical benchmark values of mean income, the income shares used for estimation, and a wide range of poverty and inequality measures.

Our empirical evaluation proceeds in two steps. First, we estimate the parametric Lorenz curves using only the ten income shares, thereby simulating a scenario where researchers are restricted to grouped data. Second, we use these estimated parametric models to analytically or numerically derive the corresponding poverty and inequality indicators. To assess the performance of the competing specifications,  we calculate the estimation error by comparing these parametric predictions directly against their true empirical counterparts.

Formally, let $\mathbf{x}_d=(x_{1d},\dots,x_{N_d})$ denote the income observations for dataset $d$, consisting of an i.i.d. sample of size $N_d$ drawn from an underlying continuous income distribution. Let $\eta_d$ denote the population value of a given poverty or inequality measure associated with the underlying distribution of dataset $d$. Suppose that only the sample mean, $\bar{x}_d$, and the grouped income shares, $\mathbf{s}_d$, are observed, following the structure described in Section~\ref{methods}. For each dataset, the estimation error of a given measure is defined as:
$$
\varepsilon_d=\hat{g}(\mathbf{s}_d,\bar{x}_d)-\eta_d,
$$
where $\hat{g}(\mathbf{s}_d,\bar{x}_d)$ denotes the estimate obtained from the fitted parametric Lorenz curve using only grouped information. The relative version of this statistic is defined as $\varepsilon_d/\eta_d$, which expresses the estimation error as a proportion of the empirical benchmark and facilitates comparisons across measures with different scales. Since the population parameter $\eta_d$ is unobservable in practice, we approximate it using its empirical counterpart computed from the underlying microdata. Under standard regularity conditions, this empirical error converges asymptotically to the population error as the sample size increases.

\subsection{Data sources and sample selection}
Our analysis relies on two complementary data sources. First, we use the Luxembourg 
Income Study (LIS) Database, the largest repository of harmonized income microdata. 
LIS contains more than one thousand datasets covering 52 countries across Europe, 
North America, Latin America, Africa, Asia, and Australasia over the last five 
decades. However, its coverage is concentrated predominantly in middle- and high-income 
economies. To evaluate the performance of alternative Lorenz curve specifications 
in estimating inequality and poverty in low-income countries, we complement 
the LIS data with records from the World Bank's PIP.

The PIP database contains 2,584 datasets whose welfare estimates are derived from 
different methodologies, including direct access to household microdata, synthetic 
distributions, grouped data, and internal imputations. To avoid confounding our 
analysis with approximation errors introduced by alternative parametric or imputation 
procedures, we restrict the sample to datasets constructed directly from individual-level 
microdata. This filtering process yields 2,374 PIP datasets. For these observations, 
the World Bank maintains access to the underlying survey data, ensuring that the 
published estimates can be treated as reliable empirical benchmarks.

A subset of these PIP observations is constructed using LIS data.\footnote{These 
datasets correspond to Australia, Canada, Germany, Israel, Japan, South Korea, 
Taiwan (China), the United Kingdom, the United States, and high-income European 
Union countries for years preceding EU-SILC.} To avoid duplicating underlying 
survey information, we exclude from PIP those datasets explicitly derived from 
LIS microdata. This exclusion ensures that our empirical evaluation is based entirely 
on a pool of mutually exclusive and independent empirical distributions, avoiding 
any artificial sample inflation that would bias our aggregate performance metrics 
toward the structural characteristics of those specific economies. After applying 
this restriction, the final analytical sample comprises 1,974 PIP datasets.

Beyond mean and median income, the PIP platform provides information on inequality 
and poverty statistics. Regarding inequality, we extract the Gini index, which 
is highly responsive to transfers around the center of the income distribution, 
and the Mean Log Deviation (MLD), which is more sensitive to redistributions 
affecting the lower tail. A more extensive catalog of poverty measures is available 
within the PIP database. In particular, this platform provides estimates of three 
indices from the Foster--Greer--Thorbecke class computed against a fixed 
poverty line of \$3.00 per day:the headcount ratio, the poverty gap, and the poverty severity index. 
\footnote{The extreme poverty line of \$3.00 per day 
is provided by default, though the platform allows users to define any fixed threshold.} 
Estimates of the Watts index are also available, which is widely regarded as the first 
distribution-sensitive poverty measure. Finally, the platform provides data on 
the societal poverty line (SPL) and its corresponding headcount ratio. Introduced 
by \citet{worldbank2018} in response to Recommendation 16 of the Atkinson Commission 
on Global Poverty \citep{worldbank2017}, the SPL combines absolute and relative 
elements of poverty. It is formally defined as:
\begin{equation}\label{spl}
z^{sp} = \max \left\{ 3.00, \, 1.30 + 0.5 \tilde{m} \right\}
\end{equation}
where $\tilde{m}$ denotes the median daily income in 2021 PPP US dollars.

While PIP provides pre-calculated indicators, the LIS database is accessed via 
a remote execution system that allows us to compute any desired statistic directly 
from the underlying microdata. To ensure consistency and comparability with the 
PIP sample, we compute both the Gini index and the MLD. We supplement this 
inequality analysis with two additional members of the Generalized Entropy (GE) 
family: the Theil index ($GE(1)$), which is equally sensitive to transfers 
at all income levels, and the $GE(1.5)$ index, which is more sensitive to redistributions 
among higher-income groups. Furthermore, we explore the capability of our proposed 
models to estimate the Atkinson index using inequality aversion parameters of 
$\epsilon = 1$ and $\epsilon = 1.5$. As $\epsilon$ increases, the 
index places progressively greater weight on the economic status of 
the poorest members of society. For poverty analysis within 
the LIS sample, given its focus on middle- and high-income countries, we compute the 
headcount ratio associated with the societal poverty line defined in Eq.~\eqref{spl}.

\begin{table}[tbp]
\caption{\label{gen_lc} Number of datasets violating the regularity conditions of the Lorenz curve}
\begin{center}
\begin{tabular}{l c c c c c c}
\toprule
Model & Ortega & Kakwani & $SCS$ & $L_3$ & GQ & No. datasets \\
\midrule
PIP	&	0	&	10	&	590	&	0	&	305	&	1,974	\\
LIS	&	0	&	7	&	322	&	0	&	109	&	1,006	\\
\bottomrule
\end{tabular}
\end{center}
\begin{tablenotes}
\item \footnotesize{Source: authors' compilation.}
\end{tablenotes}
\end{table}

The five alternative models are estimated for each of the 1,006 LIS datasets 
and 1,974 PIP datasets. Table~\ref{gen_lc} reports the number of datasets for 
which the estimated parameters fail to satisfy the regularity conditions required 
to constitute a genuine Lorenz curve. The two-parameter 
specifications (Ortega and Kakwani) generate valid Lorenz curves in virtually 
all datasets, as their simpler parametric structure imposes fewer demands on the 
data. Among the three-parameter specifications, our proposed $L_3$ curve 
satisfies the regularity conditions of a genuine Lorenz curve in every dataset 
analyzed, without a single violation. By contrast, the SCS yields non-genuine 
curves in approximately one third of the observations, and the GQ 
specification violates the regularity conpercentditions in 15 percent of PIP datasets and 
11 percent of LIS datasets. 

The perfect regularity of $L_3$ is particularly noteworthy for two 
reasons. First, among all three-parameter specifications, $L_3$ 
simultaneously achieves the best fit to the observed income shares 
(Table~\ref{rss}) and the simplest regularity conditions. Second, this result is not 
achieved by sacrificing parametric flexibility, as $L_3$ has the same 
number of effective parameters as SCS and GQ. Taken together, these 
properties suggest that $L_3$ may resolve a fundamental 
tension in the existing literature between model flexibility, empirical 
fit, and theoretical validity.

To ensure comparability across model specifications, the subsequent analysis is conducted on 
the balanced subsample of datasets for which valid estimates are jointly obtained 
for all five Lorenz curves. After applying this restriction, the final analytical 
sample consists of 1,735 datasets, including 1,146 observations from PIP and 589 from LIS.

Restricting the analysis to datasets for which all five specifications yield genuine 
Lorenz curves may introduce a selection effect if the probability of obtaining valid 
estimates is systematically associated with the characteristics of the underlying income 
distribution. To assess this possibility, we compare the empirical poverty and 
inequality indicators of the excluded datasets with those that consistently yield 
genuine curves. Tables \ref{selection_pip} and \ref{selection_lis} in Appendix B 
report the average empirical values of these indicators for both the retained and 
excluded datasets, alongside the 95 percent confidence intervals for the differences in means.

For the LIS sample, the comparison reveals only limited differences between the two 
groups. Median income, the 
societal poverty rate, and most inequality indicators do not differ 
significantly. Excluded datasets exhibit only modestly higher values 
for the $GE(1.5)$ index and slightly lower average income. In the PIP 
sample, the datasets included in the analysis tend to display slightly 
lower mean and median incomes and somewhat higher MLD indices and 
societal poverty rates. By contrast, differences in extreme poverty 
measures are generally not statistically significant. Ultimately, even 
where statistically significant differences are observed, the point 
estimates of these discrepancies are relatively small, generally 
accounting for less than 10 percent of the baseline values.

Overall, these results indicate that the balanced subsample restriction 
introduces at most modest differences in the composition of the 
analytical sample, none of which are large enough to materially affect 
the interpretation of our findings. As a robustness check, we repeat 
the evaluation using model-specific samples. That is, instead of 
restricting the baseline analysis to the balanced subsample, each model 
is evaluated using the largest available set of datasets for which that 
specific functional form satisfies the Lorenz regularity conditions. 
These results, presented in Appendix B (Tables~\ref{LIS_app1} and 
\ref{poverty_app1}), fully confirm the main conclusions of the baseline 
analysis presented in the following sections.

\subsection{Estimation of poverty measures from grouped data}
We begin the empirical evaluation by examining the performance of the alternative Lorenz curve specifications in estimating the SPR, the only poverty measure available in PIP with a sound theoretical basis for application to high-income countries, such as those predominantly featured in the LIS database. This allows us to assess model performance over the broadest possible sample, comprising 
1,735 datasets.\footnote{Results estimated separately for LIS and PIP samples 
are reported in Appendix B, Tables \ref{spr_app1} and \ref{spr_app2}, respectively.} 
Table \ref{spr} reports the mean absolute estimation error across datasets, 
computed as $|\varepsilon_d|$, to avoid positive and negative deviations offsetting 
one another. To facilitate interpretation and comparability across datasets, we 
also report the relative error expressed as a percentage of the empirical benchmark ($|\varepsilon_d / \eta_d|$).

Our results suggest that, on average, the $L_3$ curve provides the most accurate estimates of the SPR, with a mean relative error of only 2.65 percent. All alternative specifications exhibit average errors above 3 percent. Ortega's curve and the SCS deliver very similar levels of precision, with estimated poverty rates deviating from their empirical counterparts by approximately 3.5 percent on average. The GQ specification produces larger errors, close to 5 percent, while Kakwani's special case exhibits the weakest performance among the models considered.

Beyond the average performance documented in Table~\ref{spr}, it is 
important to examine the full distribution of estimation errors across 
datasets. Average errors may conceal substantial heterogeneity in model performance and provide little information about whether a specification exhibits systematic bias in one direction. Figure \ref{violin_spr} presents violin plots of the error distributions for the four best-performing models in the estimation of the SPR. These plots combine kernel density estimates with embedded boxplots, allowing for an immediate visual comparison of the location, dispersion, and asymmetry of estimation errors across specifications.

\begin{table}[h]
\caption{\label{spr}
Average absolute error in the estimation of the societal poverty rate}
\centering
\begin{tabular}{lccccc}
\toprule
 & Ortega & Kakwani & SCS & $L_3$ & GQ \\
\midrule
Average error $(\overline{|\varepsilon|}$)	&	0.0074	&	0.0174	&	0.0075	&	0.0051	&	0.0079	\\
Relative error (\%)	&	(3.46)	&	(10.05)	&	(3.48)	&	(2.65)	&	(4.98)	\\
\bottomrule
\end{tabular}
\begin{tablenotes}
\footnotesize
\item \footnotesize{These figures are computed across 1,735 datasets, comprising 1,146 from the PIP database and 589 from the LIS database. The SPR is computed via numerical integration.}
\item Source: Authors' calculations.
\end{tablenotes}
\end{table}

\begin{figure}[tbhp]
\centering
\includegraphics[width=\columnwidth]{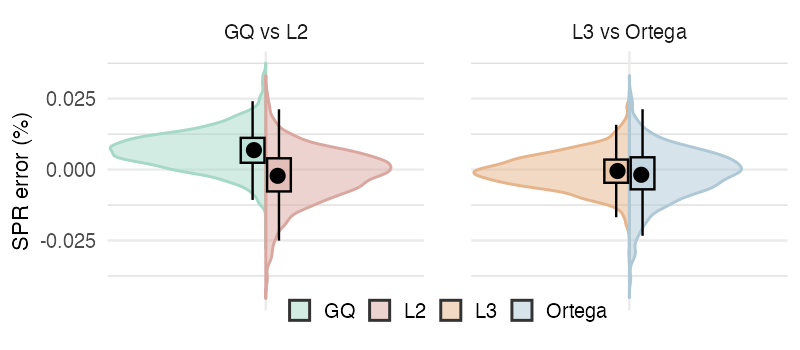}
\caption{Distribution of the error in the estimation of the SPR}
\label{violin_spr}
\end{figure}

These graphs suggest that the $L_3$ specification not only achieves the lowest average estimation error but also exhibits the narrowest error distribution. Therefore, this model exhibits greater precision and more stable performance across datasets. Ortega's specification and the SCS display broadly similar patterns, with slightly greater dispersion than $L_3$ and somewhat heavier lower tails. The GQ model shows a clear tendency to overestimate the societal poverty rate. Indeed, the GQ model overestimates the SPR in roughly 83 percent of the datasets analyzed.

The superiority of $L_3$ in estimating the SPR is robust to the sample 
restriction imposed by the balanced subsample. When each model is evaluated 
on its own largest valid sample, the $L_3$ model continues to deliver the lowest average estimation error in both the LIS 
and PIP databases (see Tables~\ref{LIS_app1} and~\ref{poverty_app1}). Moreover, 
the systematic tendency of the GQ model to overestimate the SPR documented above 
persists across the full samples. This provides additional evidence that this pattern is not sensitive to 
the particular datasets included in the balanced subsample.

We now turn to the estimation of alternative poverty statistics for the extreme poverty line, defined by the World Bank at \$3 per day. Although these poverty measures are available for all countries included in PIP, this threshold is primarily relevant for low-income settings. In middle- and high-income countries, poverty rates evaluated at this line are typically equal or very close to zero. When the empirical values of poverty measures approach zero, relative estimation errors become unstable and difficult to interpret, as even small absolute deviations may translate into disproportionately large percentage errors. Consequently, relative errors cease to provide informative comparisons of model performance across datasets. To provide a more meaningful assessment of the ability of alternative Lorenz curve specifications to estimate extreme poverty, we restrict the PIP sample to datasets for which the empirical poverty headcount ratio is at least 2 percent. This restriction yields a final sample of 655 datasets.\footnote{Summary statistics, including the mean absolute error and selected quantiles of the error distribution for the full PIP sample, are reported in Appendix B, Table \ref{poverty_app2}. These results indicate that the conclusions presented in this section are robust to the inclusion of datasets with very low poverty incidence and are not driven by the sample restriction imposed to improve the interpretability of relative estimation errors.}

\begin{table}[tbp]
\caption{\label{poverty} 
Average absolute error in the estimation of poverty measures using five alternative Lorenz curve specifications} 
\centering
\begin{tabular}{l l ccccc}
\toprule
 & & Ortega & Kakwani & SCS & $L_3$ & GQ \\
\midrule

Headcount	&	$\overline{|\varepsilon|}$	&	0.0103	&	0.0176	&	0.0103	&	0.0045	&	0.0075	\\
ratio	&	(\%)	&	(14.25)	&	(30.91)	&	(14.44)	&	(5.74)	&	(13.10)	\\
Poverty 	&	$\overline{|\varepsilon|}$	&	0.0063	&	0.0079	&	0.0064	&	0.0024	&	0.0037	\\
gap	&	(\%)	&	(33.04)	&	(42.32)	&	(34.34)	&	(14.26)	&	(25.79)	\\
Severity	&	$\overline{|\varepsilon|}$	&	0.0057	&	0.0061	&	0.0059	&	0.0024	&	0.0035	\\
index	&	(\%)	&	(61.64)	&	(52.41)	&	(64.82)	&	(27.84)	&	(37.73)	\\
Watts	&	$\overline{|\varepsilon|}$	&	0.0163	&	0.0160	&	0.0170	&	0.0068	&	0.0094	\\
index	&	(\%)	&	(45.95)	&	(46.52)	&	(48.45)	&	(20.37)	&	(31.28)	\\
\bottomrule
\end{tabular}
\begin{tablenotes}
\footnotesize
\item Notes: These figures are computed using 655 PIP datasets for which all five Lorenz curve specifications yield genuine estimates and whose empirical poverty headcount ratio exceeds 2 percent. The corresponding poverty estimates are obtained through numerical integration. 
\item Source: Authors' calculations.
\end{tablenotes}
\end{table}

Table \ref{poverty} reports the mean absolute estimation error for the four poverty measures available in PIP, together with the mean absolute relative error (reported in parentheses), which quantifies the average deviation from the empirical benchmark as a percentage. The results indicate that the $L_3$ curve consistently provides the most accurate estimates across all poverty measures considered. For the poverty headcount ratio, the average relative error is approximately 6 percent, substantially lower than that of the competing specifications. The SCS and Ortega's curve exhibit very similar performance, with relative errors close to 14 percent. 

Unlike the results obtained for the SPR, the GQ specification emerges here as the second best-performing model. Since the societal poverty line lies well above the \$3 threshold in most datasets,\footnote{The societal poverty line coincides with the extreme poverty threshold in only 5 datasets from LIS and 80 datasets from PIP.} this pattern suggests that the GQ specification provides, on average, a more accurate representation of the lower tail of the income distribution than Ortega specification and the SCS, which appear to capture the centre of the distribution more effectively.

Although the ranking of model performance remains broadly unchanged across poverty measures, the magnitude of the relative error increases markedly when moving from the headcount ratio to the poverty gap, poverty severity, and the Watts index. This pattern is expected given the construction of these indices. While the headcount ratio depends only on the proportion of individuals below the poverty line, the remaining measures incorporate information on the depth and distribution of poverty among the poor. Consequently, small approximation errors in the lower tail of the estimated Lorenz curve tend to accumulate and become amplified in these more distribution-sensitive poverty measures.

Despite the larger errors observed for these measures, the $L_3$ curve consistently 
outperforms the alternative specifications. Importantly, this result is not a 
statistical artifact driven by sample selection bias. When each model is evaluated 
on its own largest valid sample, $L_3$ continues to outperform GQ on every poverty measure considered (see Table 
\ref{poverty_app1}). For instance, the average absolute error in the poverty 
headcount ratio is 0.0040 for $L_3$ versus 0.0065 for GQ across all 1,974 PIP 
datasets. This pattern holds across all poverty measures, which suggests that the superiority of $L_3$ reflects genuine 
approximation accuracy rather than a favorable sample restriction.

Beyond comparing average performance, it is also informative to examine the full distribution of estimation errors across datasets. Figure \ref{violin_poverty} presents violin plots of the estimation errors for the poverty headcount ratio, poverty gap, and poverty severity index for the four best-performing specifications. For the headcount ratio, the error distributions appear relatively symmetric across models and are centered close to zero, providing little evidence of systematic overestimation or underestimation. Among the competing specifications, the $L_3$ curve exhibits the narrowest distribution of errors. The GQ model follows closely, while Ortega specification and the SCS display slightly greater dispersion and very similar variability.

\begin{figure}[hp]
\centering
\includegraphics[width=\columnwidth]{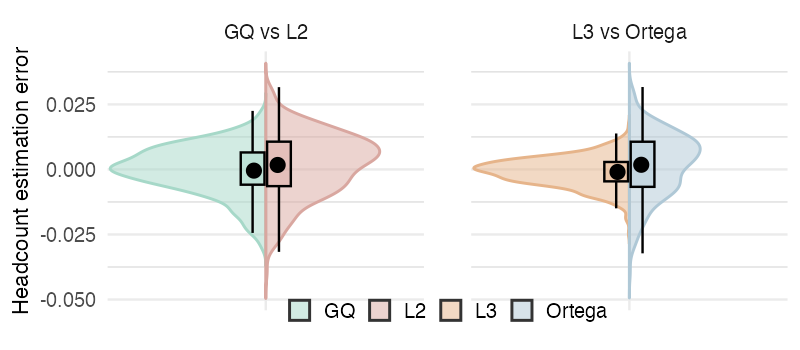}

\vspace{0.1cm}

\includegraphics[width=\columnwidth]{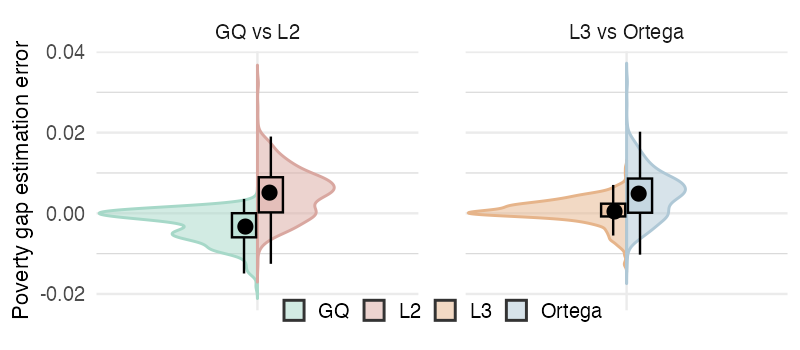}
\vspace{0.1cm}
\includegraphics[width=\columnwidth]{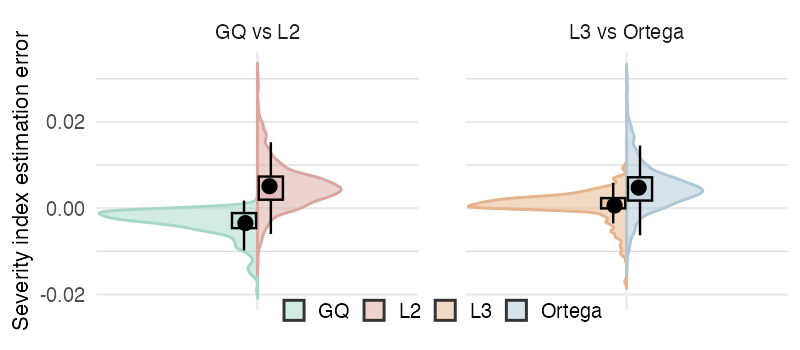}
\caption{\label{violin_poverty}
Distribution of estimation errors for poverty measures.}
\end{figure}

For the remaining poverty measures, the GQ specification exhibits a clear tendency to underestimate poverty.\footnote{This pattern is also observed for the Watts index (see Appendix B, Figure \ref{violin_watts}).} In particular, it underestimates the poverty gap in approximately 75 percent of the datasets. The downward bias is even more pronounced for the poverty severity and Watts indices, for which the GQ model produces estimates below their empirical counterparts in approximately 98 percent of the datasets. This result suggests that, although the GQ model provides a relatively accurate approximation of the incidence of poverty, it tends to underestimate the depth and intensity of deprivation among poor individuals. In other words, the lower tail of the distribution appears insufficiently represented once the poverty measures become sensitive not only to the proportion of poor individuals but also to their distance from the poverty line. This pattern  of GQ underestimation persists when each model is evaluated on its own  largest valid sample (see Table~\ref{poverty_app1}). For these distribution-sensitive poverty measures, the $L_3$ specification continues to exhibit the narrowest error distribution. Hence, this specification provides not only lower average estimation  errors but also more robust performance across heterogeneous empirical distributions.

\subsection{Estimation of inequality measures from grouped data}
We now turn to the estimation of inequality measures. Table \ref{inequality} presents the absolute errors in the estimation of the Gini index and the MLD, calculated as the absolute difference between the estimated values and the benchmark survey values reported in PIP and calculated from LIS.\footnote{The mean absolute error and selected quantiles of the error distribution computed separately for the PIP and LIS samples are reported in Appendix B, Tables \ref{ineq_app1} and \ref{ineq_app2}, respectively. These results are consistent with the conclusions presented in this section.} In line with the poverty results, the $L_3$ Lorenz curve yields the most accurate estimates for the MLD, with an average error of 0.0066 (2.75 percent). Although the GQ model also performs well, its average error of 0.0094 is about 50 percent higher than that of the $L_3$ specification. Ortega's specification and the SCS curve again show similar performances with estimation errors around 4.5 percent on average.

\begin{table}[btp]
\caption{\label{inequality}
Absolute error in the estimation of inequality measures using alternative Lorenz curve specifications}
 \begin{center}
\begin{tabular}{l l c c c c c}
\toprule												
	&			&Ortega 	&Kakwani 	&SCS 	& $L_3$ 	&GQ \\	
\midrule								MLD	&	$\overline{|\varepsilon|}$	&	0.0130	&	0.0165	&	0.0126	&	0.0066	&	0.0094	\\
	&	(\%)	&	(4.51)	&	(7.33)	&	(4.32)	&	(2.75)	&	(4.03)	\\
Gini index	&	$\overline{|\varepsilon|}$	&	0.0017	&	0.0021	&	0.0015	&	0.0008	&	0.0008	\\
	&	(\%)	&	(0.42)	&	(0.62)	&	(0.38)	&	(0.21)	&	(0.2)	\\		

\bottomrule
\end{tabular}
 \end{center}
\begin{tablenotes}
\item \footnotesize{These figures are computed across 1,735 datasets, comprising 1,146 from the PIP database and 589 from the LIS database. MLD estimates are derived using numerical integration. Gini indices are computed using Eqs. \eqref{Gini_L0}, \eqref{Gini_L1}, \eqref{Gini_L2} and \eqref{Gini_L3} for the models given in (\ref{SpecialCase1}, \ref{lco}, \ref{lc2}, \ref{lc3}) respectively. For the GQ Lorenz curve, the Gini index is computed by numerical integration.}
 \item \footnotesize{(\%) : Mean absolute percentage error, computed as  $\overline{|\hat{g}(.)-\eta\ |/ \eta} \times100$.}
\item \footnotesize{Source: Authors' calculations.}
\end{tablenotes}
\end{table}

Table~\ref{inequality} further suggests that the parametric Lorenz curves 
provide substantially more precise estimates for the Gini index than for 
the MLD. This result is expected given that the estimation procedure is 
based directly on points of the Lorenz curve, and the Gini coefficient 
itself can be expressed geometrically as twice the area between the Lorenz 
curve and the line of perfect equality. The $L_3$ and GQ specifications 
deliver very competitive results, with nearly identical average relative 
errors of 0.2 percent. This error approximately doubles for the SCS and Ortega 
specifications. Kakwani's special 
case again exhibits the largest errors, although on average it also 
provides accurate estimates of the Gini index, deviating from the 
empirical value by less than 1 percent.

\begin{figure}[tbhp]
\centering
\includegraphics[width=\columnwidth]{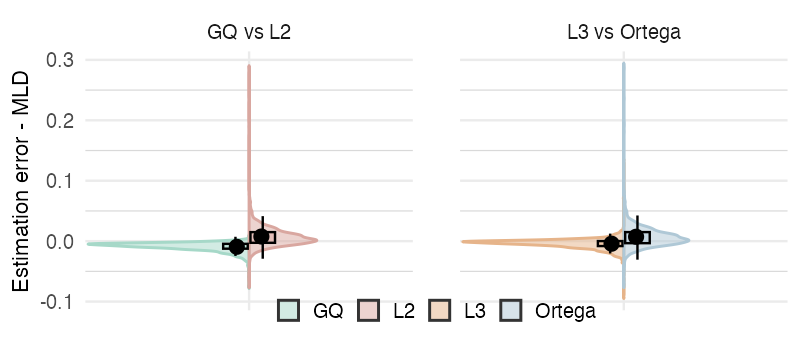}

\vspace{0.1cm}

\includegraphics[width=\columnwidth]{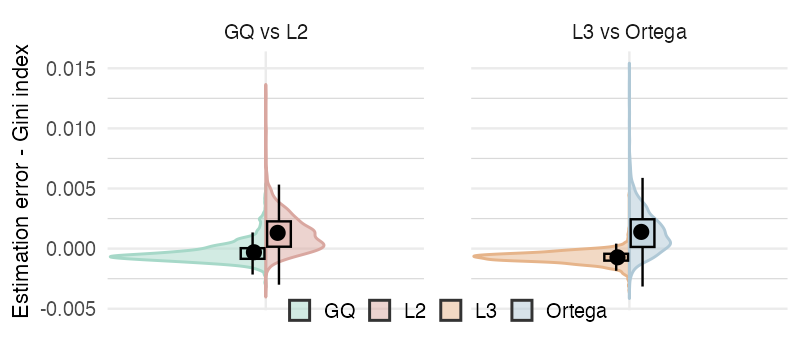}

\caption{\label{violin_ineq}
Distribution of estimation errors for inequality measures.}
\end{figure}

We now examine the full distribution of estimation errors, shown in the upper panel of Figure \ref{violin_ineq} for the MLD and in the lower panel for the Gini index. Despite its strong average performance, the GQ Lorenz curve exhibits a systematic downward bias in the estimation of the MLD, with most of the probability mass concentrated below zero. This pattern is consistent with the high sensitivity of the MLD to the lower tail of the income distribution. As discussed previously, the GQ specification tends to underrepresent the lower tail, which translates into a systematic underestimation of MLD values. The $L_3$ curve also displays a tendency to underestimate this inequality measure, although the bias appears weaker. Both models exhibit relatively low variability in their error distributions, thus suggesting stable performance across datasets. By contrast, the SCS and Ortega specifications show a tendency to overestimate inequality for both the Gini index and the MLD, as reflected by error distributions centred above zero. Moreover, these distributions are noticeably wider, which suggests greater variability in estimation performance across datasets.

These findings are robust to the sample restriction imposed by the balanced 
subsample. Tables~\ref{LIS_app1} and~\ref{poverty_app1} suggest that the 
performance ranking remains unchanged  across both databases. For the 
MLD, the average absolute error of $L_3$ is 0.0076 versus 0.0107 for GQ across 
all 1,006 LIS datasets, and 0.0060 versus 0.0077 across all 1,974 PIP datasets. 
The tendency of the GQ and $L_3$ specifications to underestimate the MLD, documented 
above for the balanced subsample, also persists across the full samples. For the Gini index, the two 
specifications remain broadly competitive across the full samples, consistent with 
the results reported in Table~\ref{inequality}.

\begin{table}[tbp]
\caption{\label{inequality_lis}
Absolute error in the estimation of alternative inequality measures using alternative Lorenz curve specifications}
 \begin{center}
\begin{tabular}{l l c c c c c}
		\toprule											
	& &	 Ortega 	&	 Kakwani 	&	 SCS 	&	 $L_3$ 	&	 GQ \\	
\midrule											
Theil	&$(\overline{|\varepsilon|}$	&0.0146	&0.0149	&0.0121	&0.0104	&0.0109	\\
 index	&(\%)		&(4.60)	&(5.56)	&(4.06)	&(4.26)	&(4.52)	\\
GE(1.5)	&$(\overline{|\varepsilon|}$	&0.0507	&0.0466	&0.0372	&0.0253	&0.0263	\\
	&	(\%)	&(11.25)	&(10.33)	&(9.09)	&(8.11)	&(8.42)	\\
Atkinson	&$(\overline{|\varepsilon|}$	&0.0067	&0.0162	&0.0066	&0.0058	&0.0091	\\
 index (1)	&(\%)	&	(3.05)	&(8.43)	&(2.94)	&(2.98)	&(4.54)	\\
Atkinson	&$(\overline{|\varepsilon|}$	&0.0160	&0.0407	&0.0159	&0.0151	&0.0248	\\
 index (1.5)	&(\%)	&	(5.01)	&(13.50)	&(4.86)	&(4.94)	&(8.00)	\\
	\bottomrule
\end{tabular}
 \end{center}
\begin{tablenotes}
\item \footnotesize{These figures are computed across 589 datasets from the LIS database. Estimates of inequality measures are computed using numerical integration.}
 \item \footnotesize{(\%) : Mean absolute percentage error, computed as  $\overline{|\hat{g}(.)-\eta\ |/ \eta} \times100$.}
\item \footnotesize{Source: Authors' calculations.}
\end{tablenotes}
\end{table}

We conclude this section by evaluating the ability of the Lorenz curve specifications to estimate additional inequality measures. Using the LIS datasets, we compute two members of the GE family: the Theil index and the GE index with parameter equal to 1.5. We also estimate two versions of the Atkinson index with inequality aversion parameters equal to 1 and 1.5.
Table \ref{inequality_lis} reports the mean absolute estimation error with its relative counterpart, which expresses the average deviation between the estimated and empirical measures as a percentage of the survey benchmark. The main finding is that the $L_3$ curve consistently delivers the lowest estimation error across all inequality measures considered.  For the Theil index, $L_3$ and SCS deliver the  lowest relative errors at approximately 4 percent, followed closely by  GQ and Ortega at around 4.5 percent. Kakwani's special case performs  less accurately, with estimation errors roughly one and a half percentage points higher on average.

\begin{figure}[tbhp]
\centering
\includegraphics[width=\columnwidth]{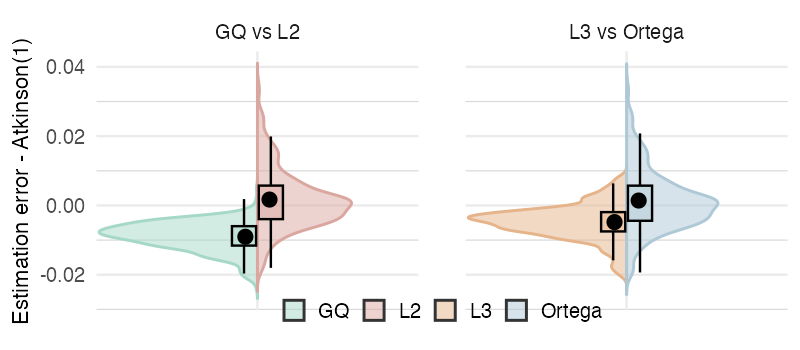}

\vspace{0.1cm}

\includegraphics[width=\columnwidth]{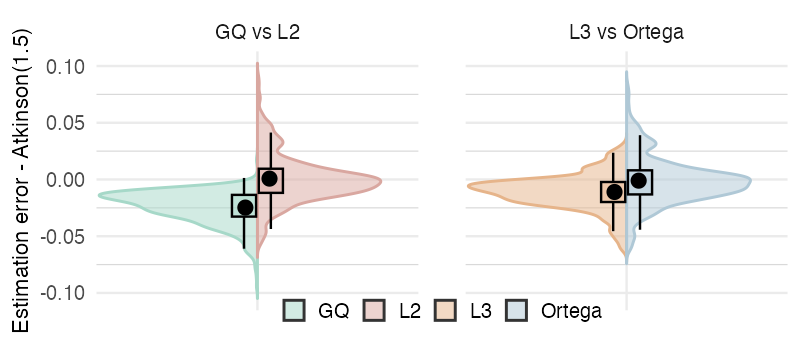}
\vspace{0.1cm}
\includegraphics[width=\columnwidth]{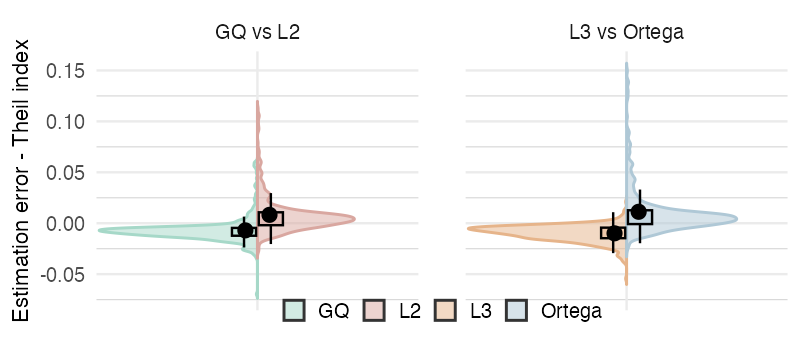}
\caption{\label{violin_ineq_lis}
Distribution of estimation errors for alternative inequality measures.}
\end{figure}


When estimating inequality measures that place greater emphasis on 
redistributions at the upper tail of the distribution, such as GE(1.5), 
estimation errors increase substantially across all specifications. The 
$L_3$ curve achieves the lowest average relative error at 8.11 percent, 
closely followed by the GQ specification (8.42 percent). The SCS curve 
exhibits a noticeably larger error of approximately 9 percent, while 
Ortega's and Kakwani's specifications report average errors close to 11 and 10 
percent respectively.

Turning to the Atkinson index, estimation errors remain comparatively 
small for both values of the inequality aversion parameter. When the 
parameter is equal to 1, the $L_3$, SCS, and Ortega specifications 
yield average relative errors of about 3 percent. Increasing the 
inequality aversion parameter to 1.5 leads to larger estimation errors, 
although these three specifications continue to display very similar 
performance, with average deviations of approximately 5 percent. By 
contrast, the GQ model performs less favorably for the Atkinson index, 
with relative errors of approximately 4.5 percent and 8 percent when 
the aversion parameter is equal to 1 and 1.5, respectively.

The full distribution of estimation errors is presented in Figure \ref{violin_ineq_lis} for the Theil index and the two Atkinson measures. The corresponding violin plots for GE(1.5) are reported in Appendix B, Figure \ref{GE15_v}. Across all inequality measures, the $L_3$ and the GQ curves emerge as the specifications with the lowest variability in estimation errors. Nevertheless, the figure reveals a tendency of $L_3$ and GQ to underestimate these inequality measures, with the bias being particularly pronounced for GQ, whose error distributions are systematically shifted below zero for these measures. 

By contrast, the SCS curve and Ortega's specification exhibit error distributions that are more closely centered around zero. However, these models display considerably greater dispersion, which implies less reliable performance across datasets. This limitation becomes especially evident for GE(1.5), where both specifications generate long upper tails and occasional very large positive estimation errors. These extreme deviations indicate that, although the SCS and Ortega models perform reasonably well on average, they are considerably less robust when applied to a heterogeneous collection of income distributions.

The performance ranking documented above is robust to the sample 
restriction imposed by the balanced subsample. When each model is 
evaluated on its own largest valid sample, the $L_3$ curve continues 
to deliver the lowest average estimation error across all measures 
considered (see Table~\ref{LIS_app1}). Notably, average estimation errors tend 
to be somewhat larger in the full samples than in the balanced subsample, 
suggesting that the datasets excluded by the balanced subsample 
restriction are on average slightly harder to approximate, consistent 
with the higher failure rates of $L_2$ and GQ that drive most of the 
exclusions.

\subsection{Sampling uncertainty of poverty and inequality estimates}

Thus far, our analysis has focused on the point-estimate accuracy of alternative Lorenz curve specifications in computing poverty and inequality measures from grouped data. However, a complete evaluation also requires examining the sampling distributions of these estimators. As discussed in Section~\ref{methods}, computing valid standard errors 
for EWMD-based estimates requires knowledge of the variance-covariance 
matrix of the income shares. Although asymptotic 
standard errors can be derived analytically under the Optimal Minimum 
Distance framework \citep{Hajargasht2016, jorda2021}, their 
implementation requires the original survey sample size and sampling 
design, information that is rarely available in secondary databases. 

In the absence of the information required to compute analytical standard errors, Monte Carlo simulation provides a practical way to characterize the sampling uncertainty of EWMD-based poverty and inequality estimates. We estimate the finite-sample bias and variability of alternative poverty and inequality measures using different sample sizes ($N=500$, $N=2,500$, and $N=5,000$). For each dataset and model, we generate $1,000$ synthetic samples and compute the grouped income shares. These synthetic income shares are then used to estimate the corresponding Lorenz curve specification from which we derive 
 poverty and inequality measures. This procedure allows us to obtain $1,000$ estimates of a given measure, that represent a nonparametric approximation of the sampling distribution, from which the bias and the standard error can be computed. The full simulation procedure is detailed in Appendix A.

Table~\ref{sd_models_b} reports the average standard deviation across 
the pooled LIS and PIP datasets for the main poverty and inequality 
measures.\footnote{All specifications exhibit relatively small 
finite-sample bias that systematically declines as sample size 
increases, consistent with asymptotic unbiasedness (see Appendix B, 
Tables \ref{bias_models_b} -- \ref{bias_models_l})}. As expected, standard errors decline 
monotonically with sample size across all specifications and measures. 
However, differences across models are particularly notable for 
inequality measures.

For the MLD and the Gini index, $L_3$ and GQ deliver the lowest 
standard errors across all sample sizes, with nearly identical 
variability between the two specifications. The remaining models exhibit 
substantially higher sampling uncertainty: at $N = 500$, the standard 
deviations of Ortega's and Kakwani's specifications are, on average, 
approximately 60 percent higher than those of $L_3$ and GQ, while the 
SCS curve yields variability estimates roughly 20 percent above those 
of $L_3$. By contrast, standard errors for the societal poverty rate 
are virtually identical across all five specifications at every sample 
size, suggesting that sampling variability is insensitive to the choice 
of parametric specification for this poverty measure.

\begin{table}[tb]
\caption{\label{sd_models_b}
Average standard errors of poverty and inequality measures using alternative Lorenz curve specifications.: LIS and PIP data}

 \begin{center}
\begin{tabular}{l c c c c c}
\toprule
	&Ortega	&Kakwani 	    & SCS	&$L_3$	&GQ\\
    	&       	&($\alpha = 1$)	    &  	&       	&       \\											\midrule\multicolumn{6}{l}{MLD}\\  \midrule											
$N = 500$	&	0.0317	&	0.0316	&	0.0238	&	0.0193	&	0.0197	\\
$N = 2500$	&	0.0154	&	0.0158	&	0.0113	&	0.0087	&	0.0094	\\
$N = 5000$	&	0.0112	&	0.0117	&	0.0082	&	0.0062	&	0.0068	\\
\midrule\multicolumn{6}{l}{Gini index}\\  \midrule											
$N = 500$	&	0.0213	&	0.0203	&	0.0169	&	0.0135	&	0.0143	\\
$N = 2500$	&	0.0106	&	0.0103	&	0.0082	&	0.0061	&	0.0067	\\
$N = 5000$	&	0.0077	&	0.0077	&	0.0060	&	0.0043	&	0.0048	\\
\midrule\multicolumn{6}{l}{SPR}\\  \midrule											
$N = 500$	&	0.0155	&	0.0154	&	0.0153	&	0.0156	&	0.0155	\\
$N = 2500$	&	0.0070	&	0.0069	&	0.0069	&	0.0069	&	0.0069	\\
$N = 5000$	&	0.0049	&	0.0049	&	0.0049	&	0.0049	&	0.0049	\\
								\bottomrule
\end{tabular}
 \end{center}
\begin{tablenotes}
\item \footnotesize{These figures are computed across 1,735 datasets, comprising 1,146 from the PIP database and 589 from the LIS database. Standard errors are estimated by Monte Carlo simulation of of $10^3$  samples of size $N = 500, 2500, 5000.$}
\item \footnotesize{Source: authors' compilation.}
\end{tablenotes}
\end{table}

\begin{table}[tbp]
\caption{\label{sd_models_p}
Average standard errors of poverty measures using alternative Lorenz curve specifications. PIP data}

 \begin{center}
\begin{tabular}{l c c c c c}
\toprule
	&Ortega	&Kakwani 	     &SCS	&$L_3$	&GQ\\
    	&       	&($\alpha = 1$)	    &  	&       	&       \\								
\midrule\multicolumn{6}{l}{Poverty headcount}\\  \midrule					$N = 500$	&	0.0140	&	0.0141	&	0.0134	&	0.0131	&	0.0140	\\
$N = 2500$	&	0.0064	&	0.0065	&	0.0061	&	0.0059	&	0.0063	\\
$N = 5000$	&	0.0045	&	0.0047	&	0.0043	&	0.0042	&	0.0045	\\
\midrule\multicolumn{6}{l}{Poverty gap}\\  \midrule											
$N = 500$	&	0.0075	&	0.0062	&	0.0072	&	0.0065	&	0.0062	\\
$N = 2500$	&	0.0035	&	0.0029	&	0.0032	&	0.0029	&	0.0028	\\
$N = 5000$	&	0.0025	&	0.0021	&	0.0023	&	0.0021	&	0.0020	\\
\midrule\multicolumn{6}{l}{Poverty severity}\\  \midrule											
$N = 500$	&	0.0053	&	0.0038	&	0.0050	&	0.0042	&	0.0038	\\
$N = 2500$	&	0.0024	&	0.0018	&	0.0022	&	0.0019	&	0.0017	\\
$N = 5000$	&	0.0018	&	0.0013	&	0.0016	&	0.0013	&	0.0012	\\
\midrule\multicolumn{6}{l}{Watts index}\\  \midrule											
$N = 500$	&	0.0148	&	0.0105	&	0.0142	&	0.0116	&	0.0102	\\
$N = 2500$	&	0.0069	&	0.0049	&	0.0064	&	0.0052	&	0.0046	\\
$N = 5000$	&	0.0049	&	0.0035	&	0.0045	&	0.0037	&	0.0032	\\
\bottomrule
\end{tabular}
 \end{center}
\begin{tablenotes}
\item \footnotesize{These figures are computed across 1,146 datasets from the PIP database. Standard errors are estimated by Monte Carlo simulation of of $10^3$  samples of size $N = 500, 2500, 5000.$}
\item \footnotesize{Source: authors' compilation.}
\end{tablenotes}
\end{table}

\begin{table}[tbp]
\caption{\label{sd_models_l}
Average standard errors of poverty and inequality measures using alternative Lorenz curve specifications. LIS data}

 \begin{center}
\begin{tabular}{l c c c c c}
\toprule
	&Ortega	&Kakwani 	    & SCS	&$L_3$	&GQ\\
    	&       	&($\alpha = 1$)	    &  	&       	&       \\								
\midrule\multicolumn{6}{l}{Atkinson index (1)}\\  \midrule					$N = 500$	&	0.0182	&	0.0174	&	0.0157	&	0.0134	&	0.0136	\\
$N = 2500$	&	0.0087	&	0.0084	&	0.0074	&	0.0061	&	0.0063	\\
$N = 5000$	&	0.0063	&	0.0061	&	0.0053	&	0.0043	&	0.0045	\\
\midrule\multicolumn{6}{l}{Atkinson index (1.5)}\\  \midrule											
$N = 500$	&	0.0214	&	0.0192	&	0.0206	&	0.0182	&	0.0177	\\
$N = 2500$	&	0.0098	&	0.0091	&	0.0093	&	0.0082	&	0.0081	\\
$N = 5000$	&	0.0069	&	0.0066	&	0.0066	&	0.0058	&	0.0058	\\
\midrule\multicolumn{6}{l}{Theil index}\\  \midrule											
$N = 500$	&	0.0506	&	0.0507	&	0.0315	&	0.0226	&	0.0245	\\
$N = 2500$	&	0.0254	&	0.0250	&	0.0156	&	0.0108	&	0.0124	\\
$N = 5000$	&	0.0187	&	0.0185	&	0.0115	&	0.0078	&	0.0091	\\
\midrule\multicolumn{6}{l}{GE(1.5)}\\  \midrule											
$N = 500$	&	0.2922	&	0.2128	&	0.1190	&	0.0426	&	0.0407	\\
$N = 2500$	&	0.1542	&	0.1210	&	0.0466	&	0.0187	&	0.0223	\\
$N = 5000$	&	0.1159	&	0.0957	&	0.0377	&	0.0135	&	0.0170	\\
\bottomrule
\end{tabular}
 \end{center}
\begin{tablenotes}
\item \footnotesize{These figures are computed across  589 datasets from the LIS database. Standard errors are estimated by Monte Carlo simulation of of $10^3$  samples of size $N = 500, 2500, 5000.$}
\item \footnotesize{Source: authors' compilation.}
\end{tablenotes}
\end{table}

This conclusion is also confirmed in Table~\ref{sd_models_p} for 
extreme poverty measures. For the poverty headcount ratio, poverty gap, 
and poverty severity index, standard errors are very similar across 
specifications, particularly at larger sample sizes. At $N = 500$, 
Ortega's specification and the SCS curve exhibit slightly higher 
variability than the remaining models, though the differences are 
modest. This pattern is more pronounced for the Watts index, where the 
standard deviations of Ortega's and SCS specifications are on average 
approximately 30 percent higher than those of $L_3$, and around 45 
percent higher than those of GQ and Kakwani's special case. Notably, 
Kakwani's special case exhibits relatively low sampling variability for 
poverty estimates despite producing the largest estimation errors among 
all specifications considered (see Table~\ref{poverty}).

Table~\ref{sd_models_l} examines the sampling variability of inequality 
measures that are sensitive to the tails of the income distribution. 
For the Atkinson indices, which place greater weight on redistributions 
among the poorest members of society, $L_3$ and GQ yield the lowest 
standard errors on average, with comparable variability between the 
two specifications. The remaining models exhibit standard errors 
approximately 20 to 30 percent higher, depending on the specification 
and sample size.

For the Theil index, $L_3$ achieves the lowest standard errors by a 
substantial margin of more than half the size of those produced by 
Ortega's specification and Kakwani's special case. The GQ specification 
displays similar variability to $L_3$ for this measure, with standard 
errors that are approximately 8.5 percent higher at $N = 500$ and 
around 15 percent higher at $N = 5{,}000$. The SCS curve falls 
between these two groups, with standard errors roughly 40 percent 
above those of $L_3$.

Finally, for GE(1.5), which is particularly sensitive to redistributions 
among the highest earners, $L_3$ and GQ display broadly similar 
variability. At $N = 500$, the standard deviation of $L_3$ is 
marginally higher than that of GQ (0.0426 versus 0.0407, a difference 
of approximately 4.6 percent). This gap reverses at larger sample 
sizes, with $L_3$ exhibiting lower variability than GQ by approximately 
12 percent at $N = 2{,}500$ and 20 percent at $N = 5{,}000$. The 
remaining specifications display substantially higher sampling 
uncertainty for this measure. Compared to the variability of the $L_3$ curve, 
standard errors are approximately six 
times larger for Ortega's specification, four times larger for 
Kakwani's special case, and twice as large for the SCS curve relative 
to $L_3$ at $N = 500$.

\section{Concluding remarks}
Despite the increasing accessibility of individual data, official income 
microdata remain unavailable in many countries. Consequently, researchers, 
practitioners, and policymakers must frequently rely on publicly available 
summary statistics of income distributions. In this context, parametric 
models serve as essential tools for reconstructing underlying distributions 
from grouped data. The Beta Lorenz curve, introduced by \citet{kakwani1980} 
has emerged as one of the most prominent specifications 
for this purpose. Its significance extends well beyond academic research: 
the model serves as a cornerstone of institutional poverty monitoring, 
being used by the World Bank to produce official global poverty estimates 
when microdata are unavailable.

In this paper, we demonstrate that Kakwani's original
proposal does not generally satisfy the 
mathematical properties of a genuine Lorenz curve. We identify the 
unique case within this family that remains theoretically consistent 
across the entire unit interval. Building on this corrected framework, 
we propose a new specification rooted in Kakwani's parametric 
framework that ensures adherence to the mathematical axioms of a 
Lorenz curve while delivering superior empirical performance. 

We evaluate the performance of this new specification through an extensive 
empirical assessment combining point-estimate accuracy and Monte Carlo 
simulations to characterize its finite-sample properties. Using more 
than 1,700 empirical datasets from the LIS and PIP databases, we 
compare its ability to recover a broad set of poverty and inequality 
measures from grouped data against that of three competing 
specifications (Ortega's specification, SCS curve, and Kakwani's 
special case), all of which can be expressed as special cases or direct 
extensions of the Beta Lorenz curve. As a primary institutional benchmark, 
we further compare our proposed specification against the GQ Lorenz 
curve.

Our results indicate that the $L_3$ specification yields genuine Lorenz 
curves in all datasets considered. Among all specifications evaluated, 
it delivers the lowest average estimation errors for both poverty and 
inequality measures, and its error distributions exhibit the lowest 
variability across datasets. These findings suggest that $L_3$ provides 
both the most accurate and the most robust estimates of poverty and 
inequality measures among the functional forms considered. The only 
potential caveat is a tendency to underestimate inequality measures, 
although this bias is substantially limited by the narrow spread of 
the error distributions.

The GQ specification emerges as the second best-performing model. 
However, it displays two important limitations relative to $L_3$. 
First, in a non-negligible share of empirical datasets, 
the estimates of GQ model are mathematically inadmissible as 
Lorenz curves. Second, even when the GQ specification 
yields a valid curve, it displays systematic deviations in the estimation 
of both poverty and inequality measures. In particular, GQ tends to 
underestimate poverty when evaluated at the extreme poverty line while 
overestimating poverty incidence at the societal poverty line. For 
inequality measures, its tendency toward underestimation is more 
pronounced than that of $L_3$, particularly for measures sensitive to 
the lower tail of the income distribution.

Taken together, these results indicate that $L_3$ achieves the most 
favorable combination of point-estimate accuracy and sampling precision 
across the full range of poverty and inequality measures considered. 
Its superiority in point-estimate accuracy is not achieved at the cost 
of greater sampling uncertainty since $L_3$ exhibits the lowest or 
comparable standard errors relative to GQ across most measures. By 
contrast, the remaining specifications are less reliable on both dimensions, particularly for 
measures that are sensitive to the tails of the income distribution, 
where their sampling variability can be several times larger than that 
of $L_3$.

Ultimately, the choice of parametric Lorenz curve specification is far 
more than a technical exercise. Practitioners have historically faced 
a trade-off between approximation accuracy, sampling precision, and 
operational reliability. The $L_3$ curve introduced in this paper 
resolves this trade-off simultaneously on all three dimensions, 
representing a substantial improvement over existing alternatives, 
including the GQ specification that currently underpins official World 
Bank poverty estimates. Our results therefore suggest that adopting 
$L_3$ in place of existing specifications would improve the accuracy 
and reliability of official poverty and inequality estimates, 
contributing to a more precise and theoretically grounded assessment 
of living standards in countries where individual data are unavailable.

 \section*{Acknowledgements}
MS and VJ acknowledge financial support from the I+D+i project (Ref. PID2024-156871NB-I00), 
financed by MICIU/AEI/10.13039/501100011033 and FEDER/UE. EGD acknowledges partial funding 
from grant PID2021-127989OB-I00, financed by the Agencia Estatal de de Investigaci\'on 
(AEI, Spain). This study uses data from the Luxembourg Income Study Database and 
the World Bank's Poverty and Inequality Platform. We thank both institutions for 
making these datasets publicly available. The authors are grateful to Daniel Gerzson Mahler, 
Christoph Lakner, Dean Jolliffe, Zurab Sajaia, Elena B\'arcena, and Raquel Sebasti\'an for their 
valuable comments and suggestions. We also thank the seminar participants at the World 
Bank's Development Data Group and the XXXIII Meeting on Public Economics for their helpful 
feedback and discussions. Any remaining errors are solely our own responsibility.

\bibliography{biblio_kakwani}

@incollection{worldbank2018,
	author = {{World Bank}},
	crossref = {worldbank2017},
	title = {Poverty and Shared Prosperity Report 2018: Piecing Together the Poverty Puzzle},
	year = {2018}}

@misc{worldbank2017,
	author = {{World Bank}},
	publisher = {Washington, DC: World Bank},
	title = {Monitoring global poverty: Report of the Commission on Global Poverty},
	year = {2017}}

@article{yitzhaki1983extension,
	author = {Yitzhaki, Shlomo},
	journal = {International Economic Review},
	pages = {617--628},
	publisher = {JSTOR},
	title = {On an extension of the {G}ini inequality index},
	year = {1983}}

@article{donaldson1980single,
	author = {Donaldson, David and Weymark, John A},
	journal = {Journal of Economic Theory},
	number = {1},
	pages = {67--86},
	publisher = {Elsevier},
	title = {A single-parameter generalization of the {G}ini indices of inequality},
	volume = {22},
	year = {1980}}

@book{arnold2018,
	author = {Arnold, Barry C and Sarabia, Jos{\'e} Mar{\'\i}a},
	publisher = {Springer},
	title = {Majorization and the Lorenz order with applications in applied mathematics and economics},
	volume = {7},
	year = {2018}}

@book{chotikapanich2008,
	author = {Chotikapanich, Duangkamon},
	publisher = {Springer Science \& Business Media},
	title = {Modeling income distributions and Lorenz curves},
	volume = {5},
	year = {2008}}

@article{Hajargasht2016,
	author = {Hajargasht, Gholamreza and Griffiths, William E},
	journal = {Econometric Reviews},
	number = {4},
	pages = {344--361},
	publisher = {Taylor \& Francis},
	title = {Minimum distance estimation of parametric {L}orenz curves based on grouped data},
	volume = {39},
	year = {2020}}

@article{mcdonald1984,
	author = {McDonald, James B.},
	copyright = {Copyright {\copyright} 1984 The Econometric Society},
	issn = {00129682},
	journal = {Econometrica},
	jstor_articletype = {research-article},
	jstor_formatteddate = {May, 1984},
	language = {English},
	number = {3},
	pages = {647-665},
	publisher = {The Econometric Society},
	title = {Some Generalized Functions for the Size Distribution of Income},
	volume = {52},
	year = {1984}}

@article{rohde2009,
	author = {Rohde, Nicholas},
	journal = {Economics Letters},
	number = {1},
	pages = {61--63},
	publisher = {Elsevier},
	title = {An alternative functional form for estimating the {L}orenz curve},
	volume = {105},
	year = {2009}}

@article{sarabia1999,
	author = {Sarabia, J-M and Castillo, Enrique and Slottje, Daniel J},
	journal = {Journal of Econometrics},
	number = {1},
	pages = {43--60},
	publisher = {Elsevier},
	title = {An ordered family of {L}orenz curves},
	volume = {91},
	year = {1999}}

@article{ryu1996,
	author = {Ryu, Hang K and Slottje, Daniel J},
	journal = {Journal of econometrics},
	number = {1-2},
	pages = {251--274},
	publisher = {Elsevier},
	title = {Two flexible functional form approaches for approximating the {L}orenz curve},
	volume = {72},
	year = {1996}}

@article{arnold1987,
	author = {Arnold, Barry C and Robertson, Christopher A and Brockett, Patrick L and Shu, Boo-Yau},
	journal = {Journal of Business \& Economic Statistics},
	number = {2},
	pages = {305--308},
	publisher = {Taylor \& Francis},
	title = {Generating ordered families of {L}orenz curves by strongly unimodal distributions},
	volume = {5},
	year = {1987}}

@article{holm1993,
	author = {Holm, Juhani},
	journal = {Journal of Econometrics},
	number = {3},
	pages = {377--389},
	publisher = {Elsevier},
	title = {Maximum entropy {L}orenz curves},
	volume = {59},
	year = {1993}}

@article{chotikapanich1993,
	author = {Chotikapanich, Duangkamon},
	journal = {Economics Letters},
	number = {2},
	pages = {129--138},
	publisher = {Elsevier},
	title = {A comparison of alternative functional forms for the {L}orenz curve},
	volume = {41},
	year = {1993}}

@article{ortega1991,
	author = {Ortega, P and Martin, G and Fernandez, A and Ladoux, M and Garcia, A},
	journal = {Review of Income and Wealth},
	number = {4},
	pages = {447--452},
	publisher = {Wiley Online Library},
	title = {A new functional form for estimating {L}orenz curves},
	volume = {37},
	year = {1991}}

@article{basmann1990,
	author = {Basmann, Robert L and Hayes, Kathy Jean and Slottje, Daniel Jonathan and Johnson, JD},
	journal = {Journal of Econometrics},
	number = {1-2},
	pages = {77--90},
	publisher = {Elsevier},
	title = {A general functional form for approximating the {L}orenz curve},
	volume = {43},
	year = {1990}}

@techreport{pakes1981,
	author = {Pakes, Ariel G},
	institution = {Department of Mathematics, University of Western Australia},
	publisher = {University of Western Australia. Department of Mathematics},
	title = {On income distributions and their {L}orenz curves},
	year = {1981}}

@techreport{wb_handbook,
	author = {Aron, Danielle V. and Castaneda Aguilar, R. Andres and Diaz-Bonilla, Carolina and Fujs, Tony H. M. J. and Garc{\'\i}a R., Diana C. and Hill, Ruth and Jularbal, Lali and Lakner, Christoph and Lara Ibarra, Gabriel and Mahler, Daniel G. and Nguyen, Minh C. and Nursamsu, Samuel and Sabatino, Carlos and Sajaia, Zurab and Seitz, William and Sjahrir, Bambang Suharnoko and Tetteh-Baah, Samuel K. and Viveros Mendoza, Martha C. and Winkler, Hernan and Wu, Haoyu and Yonzan, Nishant},
	institution = {World Bank Group},
	month = {September},
	number = {39},
	title = {September 2024 Update to the Poverty and Inequality Platform},
	type = {Technical Note},
	year = {2024}}

@article{wilson2022,
	author = {Wilson, Edgar J and Jayanthakumaran, Kankesu and Verma, Reetu},
	journal = {Review of Development Economics},
	number = {2},
	pages = {941--961},
	publisher = {Wiley Online Library},
	title = {Urban poverty, growth, and inequality: A needed paradigm shift?},
	volume = {26},
	year = {2022}}

@article{ravallion2022,
	author = {Ravallion, Martin and Chen, Shaohua},
	journal = {The Journal of Economic Inequality},
	number = {4},
	pages = {749--776},
	publisher = {Springer},
	title = {Is that really a {K}uznets curve? Turning points for income inequality in China},
	volume = {20},
	year = {2022}}

@article{kobayashi2022,
	author = {Kobayashi, Genya and Yamauchi, Yuta and Kakamu, Kazuhiko and Kawakubo, Yuki and Sugasawa, Shonosuke},
	journal = {Journal of Business \& Economic Statistics},
	number = {2},
	pages = {897--912},
	publisher = {Taylor \& Francis},
	title = {Bayesian approach to {L}orenz curve using time series grouped data},
	volume = {40},
	year = {2022}}

@article{bresson2009,
	author = {Bresson, Florent},
	journal = {Review of Income and Wealth},
	number = {2},
	pages = {266--302},
	publisher = {Wiley Online Library},
	title = {On the estimation of growth and inequality elasticities of poverty with grouped data},
	volume = {55},
	year = {2009}}

@article{minoiu2009,
	author = {Minoiu, Camelia and Reddy, Sanjay G},
	journal = {Journal of Income Distribution},
	number = {2},
	pages = {160-178},
	title = {Estimating poverty and inequality from grouped data: How well do parametric methods perform?},
	volume = {18},
	year = {2009}}

@article{villasenor1989,
	author = {Villase{\~n}or, Jos{\'e}A and Arnold, Barry C},
	journal = {Journal of Econometrics},
	number = {2},
	pages = {327--338},
	publisher = {Elsevier},
	title = {Elliptical {L}orenz curves},
	volume = {40},
	year = {1989}}

@article{kakwani1980,
	author = {Kakwani, Nanak},
	journal = {Econometrica: Journal of the Econometric Society},
	pages = {437--446},
	publisher = {JSTOR},
	title = {On a class of poverty measures},
	volume = {48},
	year = {1980}}

@incollection{arnold2008,
	author = {Arnold, Barry C},
	booktitle = {Modeling income distributions and Lorenz curves},
	editor = {Chotikapanich, Duangkamon},
	pages = {119--145},
	publisher = {Springer},
	title = {Pareto and generalized {P}areto distributions},
	year = {2008}}

@article{blanchet2022,
	author = {Blanchet, Thomas and Fournier, Juliette and Piketty, Thomas},
	journal = {Review of Income and Wealth},
	number = {1},
	pages = {263--288},
	publisher = {Wiley Online Library},
	title = {Generalized {P}areto curves: theory and applications},
	volume = {68},
	year = {2022}}

@article{jorda2021,
	author = {Jorda, Vanesa and Sarabia, Jos{\'e} Mar{\'\i}a and J{\"a}ntti, Markus},
	journal = {Journal of the Royal Statistical Society Series A: Statistics in Society},
	number = {3},
	pages = {964--984},
	publisher = {Oxford University Press},
	title = {Inequality measurement with grouped data: Parametric and non-parametric methods},
	volume = {184},
	year = {2021}}

@article{kakwani1976,
	author = {Kakwani, NC and Podder, N},
	journal = {Econometrica: Journal of the Econometric Society},
	pages = {137--148},
	publisher = {JSTOR},
	title = {Efficient Estimation of the {L}orenz Curve and Associated Inequality Measures from Grouped Observations},
	volume = {44},
	year = {1976}}

@article{rasche1980,
	author = {Rasche, Robert H and Gaffney, Janice and Koo, Anthony YC and Obst, Norman},
	journal = {Econometrica},
	number = {4},
	title = {Functional forms for estimating the {L}orenz curve.},
	volume = {48},
	year = {1980}}

@article{sitthiyot2021simple,
	author = {Sitthiyot, Thitithep and Holasut, Kanyarat},
	journal = {Humanities and Social Sciences Communications},
	number = {1},
	pages = {268},
	publisher = {Palgrave},
	title = {A simple method for estimating the {L}orenz curve},
	volume = {8},
	year = {2021}}

@article{shen2024regression,
	author = {Shen, Xiaobo and Dai, Pingsheng},
	journal = {Humanities and Social Sciences Communications},
	number = {1},
	pages = {1--8},
	publisher = {Palgrave},
	title = {A regression method for estimating {G}ini index by decile},
	volume = {11},
	year = {2024}}

@article{cheong2002empirical,
	author = {Cheong, Kwang Soo},
	journal = {Applied Economics Letters},
	number = {3},
	pages = {171--176},
	publisher = {Taylor \& Francis},
	title = {An empirical comparison of alternative functional forms for the {L}orenz curve},
	volume = {9},
	year = {2002}}

@article{schader1994fitting,
	author = {Schader, Martin and Schmid, Friedrich},
	journal = {Empirical Economics},
	number = {3},
	pages = {361--370},
	publisher = {Springer},
	title = {Fitting parametric {L}orenz curves to grouped income distributions--a critical note},
	volume = {19},
	year = {1994}}

\newpage
\section*{Appendix A}
\setcounter{table}{0}
\setcounter{figure}{0}
\setcounter{equation}{0}
\renewcommand{\theequation}{A\arabic{equation}}
\renewcommand{\thetable}{B\arabic{table}}
\renewcommand{\thefigure}{B\arabic{figure}}

\subsection*{Proof of Theorem 2.2}
\begin{proof}
The proof is based on Theorem \ref{theorem1}. Conditions $L(0)=0$ and $L(1)=1$ are satisfied for (\ref{kakwaniLC}). Now, since (\ref{kakwaniLC}) is differentiable we have,
\begin{equation*}
L'(p; a, \alpha, \beta) = 1-a p^{\alpha -1} (1-p)^{\beta -1} (\alpha -\alpha p+\beta  p).
\end{equation*}
Now we distinguish three cases,\\
Case 1: $\alpha < 1$, then,
$$
\lim_{p \to 0^+} L'(p; a, \alpha, \beta) = -\infty
$$
and the third condition in (\ref{LCgenuine}) it is not satisfied.\\
Case 2: $\alpha = 1$, then
$$
L'(0^+; a, \alpha, \beta) = 1 - a,
$$
and then $a \leq 1$.\\
Case 3: $\alpha > 1$,
$$
L'(0^+; a, \alpha, \beta) = 1
$$
and the condition is verified. Consequently, we have two possible situations,
\begin{equation}\label{RR1}
\alpha = 1\;\;\mbox{and}\;\;0\le a\le 1,
\end{equation}
or
\begin{equation}\label{RR2}
\alpha > 1.
\end{equation}
Let's see that case (\ref{RR1}) leads to model given in the Theorem statement. Since $\alpha=1$, (\ref{kakwaniLC}) is of the form (\ref{SpecialCase1}). Then, we prove now that (\ref{SpecialCase1}) is a genuine Lorenz curve if $0\le a,\beta\le 1$. First, $L_1(0)=0$ and $L_1(1)=1$. Now,
$$
L'_1(p)=1-a(1-p)^{\beta-1}(1-p+p\beta),
$$
and $L'_1(0)=1-a\ge 0$, in consequence $a\le 1$. The second derivative is,
$$
L''_1(p)=a\beta(1-p)^{\beta-2}(2-p-\beta p),
$$
and it is direct that $L''_1(p)\ge 0$ iff $\beta\le 1$ and this completes the result.\\

We continue with the case corresponding to (\ref{RR2}). Then, the second derivative of (\ref{kakwaniLC}) is given by,
\begin{equation*}
L''(p; a, \alpha, \beta) = \frac{a(1-p)^\beta p^{\alpha-2} u(p)}{(1-p)^2},
\end{equation*}
where,
\begin{equation*}
u(p) = a_2 p^2 + a_1 p + a_0
\end{equation*}
and
\begin{eqnarray*}
a_2 &=& \alpha + \beta - (\alpha + \beta)^2, \\
a_1 &=& 2\alpha\beta + 2\alpha^2 - 2\alpha, \\
a_0 &=& \alpha - \alpha^2,\\
\Delta&=&4\alpha\beta(\alpha+\beta-1),
\end{eqnarray*}
where $\Delta$ is the discriminant of the equation $u(p)=0$. Then $L''(p)\ge 0$ is verified if and only if $u(p)\geq 0$ with $p \in [0,1]$. We distinguish two cases: complex roots and real roots.

Case 1, complex roots. In this case by considering the discriminant, we obtain the conditions,
$$
\alpha+\beta<1,\;\;0<\alpha<1,
$$
which are not possible, since they contradict (\ref{RR2}).

Case 2: real roots. Real roots and $u(p)$ presents a minimum. If $p_0=\frac{\alpha}{\alpha+\beta}$ is the critical point of $u(p)$, it must satisfy $p_0>0$ and $u''(p_0)>0$, which leads to the conditions
$$\alpha>0,\;\;0<\alpha+\beta<1,$$
which contradict (\ref{RR2}). Finally, $u(p)$ presents a maximum and must be a positive function. Then, $u(0)\ge 0$, $u(1)\ge 1$, $p_0>0$ and $u''(p_0)>0$, which leads to
$$0<\alpha,\beta\le 1,\;\;\alpha+\beta>1,$$
which again contradict (\ref{RR2}).\\
Border cases can be easily verified.
In conclusion, the only genuine Lorenz curve based on Kakwani's proposal given in (\ref{kakwaniLC}), is the curve (\ref{SpecialCase1}), together with the parameter restrictions included in Theorem \ref{theorem2}.
\end{proof}

\subsection*{Additional properties of the special case of Kakwani's Lorenz curve}
A relevant generalization of the Gini index was considered by \citet{donaldson1980single}, \citet{kakwani1980} and \citet{yitzhaki1983extension}. These authors proposed the generalized Gini index defined as,
\begin{equation}\label{genGini}
G(\nu)=1-\mu(\nu+1)\int_{0}^{1}(1-p)^{\nu-1}L_X(p)dp,
\end{equation}
where $\nu\ge 1$ and $L_X(p)$ is the Lorenz curve. If we set $\nu=1$ in (\ref{genGini}) we obtain the usual Gini index. For the family (\ref{SpecialCase1}) we have,
$$
G(\nu)=1-\nu(\nu+1)\left(\frac{1}{\nu(\nu+1)}-\frac{a}{(\beta+\nu)(\beta+\nu+1)}\right).
$$
The quantile function associated with the corresponding income distribution with Lorenz curve (\ref{SpecialCase1}) is
$$
Q(p;a,\beta,\mu)=\mu\left(1-a(1-p)^\beta+a\beta p(1-p^{\beta-1}\right),
$$
where $\mu$ is the income mean and $0\le p\le 1$. Note that the support of the income distribution is $[\mu(1-a),\infty)$.

\subsection*{Monte Carlo simulation procedure}

For each dataset $d$, model $m$, and sample size $N \in \{500, 2{,}500, 
5{,}000\}$, the Monte Carlo procedure proceeds as follows:

\begin{enumerate}
\item Using the estimated model $L_m(p; \hat{\boldsymbol{\theta}}_{d,m})$, 
simulate a synthetic sample of size $N$, 
$\mathbf{x}^{*(r)}_{d,m} = (x^{*(r)}_1, \dots, x^{*(r)}_N)$.
\item Compute the grouped income shares $\mathbf{s}^{*(r)}_{d,m}$ and 
mean income $\bar{x}^{*(r)}_{d,m}$ from this synthetic sample, following 
the same grouping structure (deciles) used in the main analysis.
\item Re-estimate the parameter vector of Lorenz curve model $m$, 
$\hat{\boldsymbol{\theta}}^{*(r)}_{d,m}$, using the EWMD estimator applied 
to $\mathbf{s}^{*(r)}_{d,m}$, following the procedure described in 
Section~\ref{methods}.
\item From the fitted specification $L_m(p; \hat{\boldsymbol{\theta}}^{*(r)}_{d,m})$ 
and using the synthetic mean income $\bar{x}^{*(r)}_{d,m}$, compute the 
poverty and inequality measures of interest via numerical integration 
or the relevant closed-form expression, denoted 
$\hat{g}^{(r)}_{d,m}$.
\item Repeat steps 1--4 for $r = 1, \dots, R$, with $R = 1{,}000$.
\end{enumerate}

For a given poverty or inequality measure, this procedure yields $R$ 
estimates $\{\hat{g}^{(r)}_{d,m}\}_{r=1}^{R}$ for dataset $d$ and model 
$m$. This set of estimates provides a nonparametric approximation of the 
sampling distribution of the estimator for that particular measure. The 
finite-sample bias is computed as
\begin{equation}
b_{d,m} = \frac{1}{R}\sum_{r=1}^{R} \hat{g}^{(r)}_{d,m} - \hat{g}_{d,m},
\end{equation}
where $\hat{g}_{d,m}$ denotes the estimate obtained using the original 
grouped income shares $\mathbf{s}_d$, and the standard error is computed 
as
\begin{equation}
SE_{d,m} = \sqrt{\frac{1}{R-1}\sum_{r=1}^{R}\left(\hat{g}^{(r)}_{d,m} - \tilde{g}_{d,m}
\right)^2}, 
\end{equation}
where $\tilde{g}_{d,m}=\frac{1}{R}\sum_{r=1}^{R}\hat{g}^{(r)}_{d,m}$.

\newpage
\section*{Appendix B}

\begin{table}[hp]
\caption{\label{rss}
Godness-of-fit of alternative Lorenz curves: Residual sum of squares}
 \begin{center}
\begin{tabular}{l c c c c c}
\toprule											
	& Ortega 	& Kakwani 	&	 SCS 	&	 $L_3$ 	&	 GQ \\	
\midrule								
Average	&0.0481	&0.0928	&0.0450	&	0.0047	&	0.0079	\\
10\%	&0.0032	&0.0176	&0.0023	&	0.0005	&	0.0011	\\
25\%	&0.0091	&0.0311	&0.0079	&	0.0010	&	0.0021	\\
50\%	&0.0264	&0.0631	&0.0245	&	0.0026	&	0.0041	\\
75\%	&0.0666	&0.1435	&0.0632	&	0.0060	&	0.0089	\\
90\%	&0.1228	&0.2057	&0.1144	&	0.0111	&	0.0180	\\\bottomrule
\end{tabular}
 \end{center}
\begin{tablenotes}
\item \footnotesize{These figures are multiplied by $10^3$ and computed across 1,735 datasets, comprising 1,146 from the PIP database and 589 from the LIS database.}
\item \footnotesize{Source: authors' compilation.}
\end{tablenotes}
\end{table}

\begin{table}[h]
\caption{\label{selection_pip}
Comparison of average empirical poverty and inequality measures between included and excluded datasets in the PIP sample}
\centering
\begin{tabular}{lccccc}
\toprule
Statistic &
Included &
Excluded &
Difference &
95\% CI  \\
\midrule
Mean income	&22.5783	&25.1648	&	-2.5864	&	(-4.5076, -0.6653)	\\
Median income	&18.0349	&20.6452	&	-2.6102	&	(-4.3005, -0.9200)	\\
Gini index	&0.2733	&0.2681	&	0.0052	&	(-0.0086, 0.0190)	\\
MLD	&0.3858	&0.3695	&	0.0163	&	(0.0078, 0.0248)	\\
SPR	&0.2682	&0.2413	&	0.027	&	(0.0130, 0.0409)	\\
Headcount ratio&0.1257	&0.1135	&	0.0123	&	(-0.0049, 0.0294)	\\
Poverty gap	&0.0468	&0.0458	&	0.001	&	(-0.0069, 0.0090)	\\
Poverty severity&0.0245	&0.0258	&	-0.0013	&	(-0.0061, 0.0036)	\\
Watts index	&0.0714	&0.0738	&	-0.0023	&	(-0.0155, 0.0108)	\\
\bottomrule
\end{tabular}

\begin{tablenotes}
\footnotesize
\item Notes: Included datasets correspond to observations retained in the balanced sample, i.e., datasets for which all five Lorenz curve specifications yield genuine estimates. Excluded datasets correspond to observations removed by this restriction. The last two columns report the lower and upper bounds of the 95 percent confidence interval for the difference in means.
\item Source: Authors' calculations.
\end{tablenotes}
\end{table}

\begin{table}[tbp]
\caption{\label{selection_lis}
Comparison of average empirical poverty and inequality measures between included and excluded datasets in the LIS sample}
\centering
\begin{tabular}{lcccc}
\toprule
Statistic &
Included &
Excluded &
Difference &
95\% CI \\
\midrule
Mean income	&39.6545	&36.6071	&3.0474	&(0.3725, 5.7223)		 \\
Median income	&32.6660	&30.6878	&1.9781	&	(-0.3807, 4.3369)		 \\
Gini index	&0.3635	&0.3608	&0.0028	&	(-0.0090, 0.0146)		 \\
MLD (GE(0))	&0.2529	&0.2607	&-0.0079	&	(-0.0280, 0.0123)		 \\
Theil index (GE(1))&0.2513	&0.2722	&-0.0209	&	(-0.0423, 0.0005)		 \\
GE(1.5)	&0.2950	&0.3489	&-0.0539	&	(-0.0904, -0.0174)		 \\
Atkinson ($\varepsilon=1$)	&0.2181	&0.2177	&0.0004	&(-0.0134, 0.0143)		 \\
Atkinson ($\varepsilon=1.5$)	&0.3188	&0.3090	&0.0099	&(-0.0083, 0.0281)		 \\
Societal poverty rate 	&0.1893	 &0.1828	 &0.0064	 &(-0.0067, 0.0196)		 \\
\bottomrule
\end{tabular}

\begin{tablenotes}
\footnotesize
\item Notes: Included datasets correspond to observations retained in the balanced sample, i.e., datasets for which all five Lorenz curve specifications yield genuine estimates. Excluded datasets correspond to observations removed by this restriction. The last column reports the 95 percent confidence interval for the difference in means.
\item Source: Authors' calculations.
\end{tablenotes}
\end{table}

\begin{table}[tbp]
\caption{\label{LIS_app1}
Absolute error in the estimation of the poverty and inequality measures using alternative Lorenz curve specifications using all the datasets available in LIS}
 \begin{center}
\begin{tabular}{l c c c c c}
\toprule											
	&	 Ortega 	&	 Kakwani 	&	 SCS 	&	 $L_3$ 	&	 GQ \\	
\midrule											
\multicolumn{6}{l}{SPR}\\  \midrule											
Average	&	0.0060	&	0.0169	&	0.0060	&	0.0047	&	0.0087	\\
10\%	&	-0.0088	&	-0.0149	&	-0.0093	&	-0.0068	&	-0.0003	\\
50\%	&	0.0012	&	0.0121	&	0.0013	&	0.0004	&	0.0078	\\
90\%	&	0.0102	&	0.0275	&	0.0102	&	0.0079	&	0.0160	\\
\midrule\multicolumn{6}{l}{Gini index}\\  \midrule											
Average	&	0.0013	&	0.0026	&	0.0011	&	0.0008	&	0.0009	\\
10\%	&	-0.0009	&	-0.0040	&	-0.0005	&	-0.0014	&	-0.0013	\\
50\%	&	0.0006	&	-0.0025	&	0.0005	&	-0.0008	&	-0.0007	\\
90\%	&	0.0029	&	0.0010	&	0.0025	&	-0.0002	&	0.0008	\\
\midrule\multicolumn{6}{l}{MLD}\\  \midrule											
Average	&	0.0099	&	0.0197	&	0.0093	&	0.0076	&	0.0107	\\
10\%	&	-0.0106	&	-0.0327	&	-0.0083	&	-0.0149	&	-0.0190	\\
50\%	&	0.0012	&	-0.0179	&	0.0014	&	-0.0047	&	-0.0091	\\
90\%	&	0.0197	&	-0.0078	&	0.0196	&	0.0020	&	-0.0040	\\
\midrule\multicolumn{6}{l}{Theil index}\\  \midrule											
Average	&	0.0177	&	0.0179	&	0.0143	&	0.0112	&	0.0213	\\
10\%	&	-0.0077	&	-0.0183	&	-0.0051	&	-0.0198	&	-0.0246	\\
50\%	&	0.0042	&	-0.0068	&	0.0049	&	-0.0080	&	-0.0088	\\
90\%	&	0.0539	&	0.0391	&	0.0372	&	-0.0012	&	0.0009	\\
\midrule\multicolumn{6}{l}{GE(1.5)}\\  \midrule											
Average	&	0.0734	&	0.0657	&	0.0541	&	0.0315	&	0.0490	\\
10\%	&	-0.0144	&	-0.0265	&	-0.0103	&	-0.0570	&	-0.0592	\\
50\%	&	0.0097	&	-0.0043	&	0.0119	&	-0.0165	&	-0.0167	\\
90\%	&	0.2402	&	0.2169	&	0.1540	&	0.0005	&	0.0035	\\
\midrule\multicolumn{6}{l}{Atkinson index (1)}\\  \midrule											
Average	&	0.0067	&	0.0154	&	0.0062	&	0.0056	&	0.0081	\\
10\%	&	-0.0086	&	-0.0252	&	-0.0069	&	-0.0111	&	-0.0140	\\
50\%	&	0.0010	&	-0.0144	&	0.0012	&	-0.0038	&	-0.0074	\\
90\%	&	0.0125	&	-0.0064	&	0.0124	&	0.0015	&	-0.0033	\\
\midrule\multicolumn{6}{l}{Atkinson index (1.5)}\\  \midrule											
Average	&	0.0149	&	0.0384	&	0.0146	&	0.0141	&	0.0229	\\
10\%	&	-0.0240	&	-0.0652	&	-0.0208	&	-0.0304	&	-0.0430	\\
50\%	&	-0.0011	&	-0.0345	&	-0.0004	&	-0.0078	&	-0.0184	\\
90\%	&	0.0239	&	-0.0153	&	0.0260	&	0.0068	&	-0.0080	\\
\bottomrule
\end{tabular}
 \end{center}
\begin{tablenotes}
\item \footnotesize{These figures are computed across 1,006 LIS datasets.}
\item \footnotesize{Source: authors' compilation.}
\end{tablenotes}
\end{table}

\begin{table}[tbp]
\caption{\label{poverty_app1}
Absolute error in the estimation of inequality and poverty measures using alternative Lorenz curve specifications using all the datasets available in PIP}
 \begin{center}
\begin{tabular}{l c c c c c}
\toprule
	&Ortega	&Kakwani 	     &	SCS	&	L3	&GQ\\
    &       &($\alpha = 1$)&  &       &       \\									
\midrule\multicolumn{6}{l}{Poverty headcount}\\  \midrule				
Average	&	0.0074	&	0.0121	&	0.0073	&	0.0040	&	0.0065	\\
10\%	&	-0.0089	&	-0.0278	&	-0.0086	&	-0.0069	&	-0.0127	\\
50\%	&	0.0005	&	-0.0034	&	0.0005	&	-0.0003	&	-0.0015	\\
90\%	&	0.0139	&	0.0149	&	0.0136	&	0.0054	&	0.0095	\\
\midrule\multicolumn{6}{l}{Poverty gap}\\  \midrule											
Average	&	0.0043	&	0.0055	&	0.0043	&	0.0021	&	0.0029	\\
10\%	&	-0.0037	&	-0.0140	&	-0.0035	&	-0.0041	&	-0.0075	\\
50\%	&	0.0008	&	-0.0025	&	0.0008	&	0.0000	&	-0.0015	\\
90\%	&	0.0098	&	0.0010	&	0.0100	&	0.0030	&	0.0004	\\
\midrule\multicolumn{6}{l}{Poverty severity}\\  \midrule											
Average	&	0.0037	&	0.0043	&	0.0037	&	0.0020	&	0.0026	\\
10\%	&	-0.0029	&	-0.0109	&	-0.0027	&	-0.0038	&	-0.0069	\\
50\%	&	0.0008	&	-0.0019	&	0.0008	&	0.0001	&	-0.0015	\\
90\%	&	0.0080	&	0.0000	&	0.0084	&	0.0026	&	0.0000	\\
\midrule\multicolumn{6}{l}{Watts index}\\  \midrule											
Average	&	0.0108	&	0.0116	&	0.0111	&	0.0055	&	0.0071	\\
10\%	&	-0.0077	&	-0.0275	&	-0.0070	&	-0.0099	&	-0.0186	\\
50\%	&	0.0029	&	-0.0063	&	0.0034	&	0.0004	&	-0.0045	\\
90\%	&	0.0228	&	-0.0001	&	0.0247	&	0.0069	&	-0.0002	\\
\midrule\multicolumn{6}{l}{Gini index}\\  \midrule											
Average	&	0.0019	&	0.0021	&	0.0016	&	0.0007	&	0.0009	\\
10\%	&	-0.0007	&	-0.0037	&	-0.0002	&	-0.0013	&	-0.0010	\\
50\%	&	0.0014	&	-0.0013	&	0.0013	&	-0.0006	&	-0.0003	\\
90\%	&	0.0039	&	0.0019	&	0.0035	&	0.0000	&	0.0016	\\
\midrule\multicolumn{6}{l}{MLD}\\  \midrule											
Average	&	0.0143	&	0.0145	&	0.0137	&	0.0060	&	0.0077	\\
10\%	&	-0.0107	&	-0.0317	&	-0.0092	&	-0.0150	&	-0.0166	\\
50\%	&	0.0067	&	-0.0098	&	0.0068	&	-0.0012	&	-0.0050	\\
90\%	&	0.0284	&	0.0000	&	0.0283	&	0.0040	&	0.0000	\\
\midrule\multicolumn{6}{l}{SPR}\\  \midrule											
Average	&	0.0076	&	0.0176	&	0.0076	&	0.0052	&	0.0076	\\
10\%	&	-0.0145	&	-0.0125	&	-0.0150	&	-0.0088	&	-0.0022	\\
50\%	&	-0.0031	&	0.0153	&	-0.0030	&	-0.0019	&	0.0062	\\
90\%	&	0.0074	&	0.0278	&	0.0072	&	0.0068	&	0.0145	\\

\bottomrule
\end{tabular}
 \end{center}
\begin{tablenotes}
\item \footnotesize{These figures are computed across 1,974 PIP datasets. The corresponding poverty estimates are derived via numerical integration.}
\item \footnotesize{Source: authors' compilation.}
\end{tablenotes}
\end{table}

\begin{figure}[tbhp]
\centering
\includegraphics[width=\columnwidth]{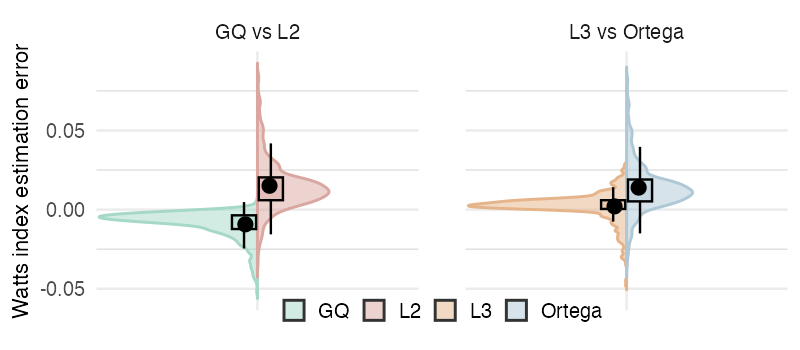}
\caption{Distribution of the error in the estimation of the Watts index}
\label{violin_watts}
\end{figure}

\begin{table}[tbp]
\caption{\label{spr_app1}
Average absolute error in the estimation of the societal poverty rate. LIS datasets}
\centering
\begin{tabular}{lccccc}
\toprule
 & Ortega & Kakwani & SCS & $L_3$ & GQ \\
\midrule
Average	&	0.0060	&	0.0169	&	0.0060	&	0.0047	&	0.0087	\\
10\%	&	-0.0088	&	-0.0149	&	-0.0093	&	-0.0068	&	-0.0003	\\
25\%	&	-0.0031	&	0.0015	&	-0.0030	&	-0.0032	&	0.0037	\\
50\%	&	0.0012	&	0.0121	&	0.0013	&	0.0004	&	0.0078	\\
75\%	&	0.0058	&	0.0209	&	0.0058	&	0.0042	&	0.0117	\\
90\%	&	0.0102	&	0.0275	&	0.0102	&	0.0079	&	0.0160	\\
\bottomrule
\end{tabular}
\begin{tablenotes}
\footnotesize
\item \footnotesize{These figures are computed across 589 datasets from the LIS database. The SPR is computed via numerical integration.}
\item Source: Authors' calculations.
\end{tablenotes}
\end{table}

\begin{table}[tbp]
\caption{\label{spr_app2}
Average absolute error in the estimation of the societal poverty rate. PIP datasets}
\centering
\begin{tabular}{lccccc}
\toprule
 & Ortega & Kakwani & SCS & $L_3$ & GQ \\
\midrule
Average	&	0.0084	&	0.0172	&	0.0086	&	0.0052	&	0.0074	\\
10\%	&	-0.0176	&	-0.0087	&	-0.0183	&	-0.0092	&	-0.0024	\\
25\%	&	-0.0097	&	0.0048	&	-0.0103	&	-0.0058	&	0.0013	\\
50\%	&	-0.0040	&	0.0159	&	-0.0046	&	-0.0018	&	0.0057	\\
75\%	&	0.0027	&	0.0220	&	0.0022	&	0.0029	&	0.0105	\\
90\%	&	0.0073	&	0.0270	&	0.0064	&	0.0067	&	0.0147	\\
\bottomrule
\end{tabular}
\begin{tablenotes}
\footnotesize
\item \footnotesize{These figures are computed across 1,146 datasets from the PIP database. The SPR is computed via numerical integration.}
\item Source: Authors' calculations.
\end{tablenotes}
\end{table}

\begin{table}[tbp]
\caption{\label{poverty_app2}
Absolute error in the estimation of the poverty measures using only PIP datasets with genuine Lorenz curves for the five alternative specifications}
 \begin{center}
\begin{tabular}{l c c c c c}
\toprule
	&Ortega	&	Kakwani ($\epsilon = 1$)	&	SCS	&	L3	&	GQ	\\
\midrule\multicolumn{6}{l}{Poverty headcount}\\  \midrule											
Average	&0.0091	&	0.0131	&	0.0092	&	0.0044	&	0.0069	\\
10\%	&-0.0106	&	-0.0299	&	-0.0110	&	-0.0075	&	-0.0139	\\
25\%	&-0.0043	&	-0.0147	&	-0.0040	&	-0.0042	&	-0.0083	\\
50\%	&0.0027	&	-0.0049	&	0.0027	&	-0.0004	&	-0.0023	\\
75\%	&0.0105	&	0.0011	&	0.0106	&	0.0034	&	0.0021	\\
90\%	&0.0167	&	0.0148	&	0.0169	&	0.0062	&	0.0080	\\
\midrule\multicolumn{6}{l}{Poverty gap}\\  \midrule											
Average	&	0.0087	&	0.0128	&	0.0087	&	0.0043	&	0.0069	\\
10\%	&	-0.0108	&	-0.0292	&	-0.0107	&	-0.0075	&	-0.0141	\\
50\%	&	0.0011	&	-0.0045	&	0.0009	&	-0.0005	&	-0.0023	\\
90\%	&	0.0160	&	0.0148	&	0.0164	&	0.0058	&	0.0084	\\
\midrule\multicolumn{6}{l}{Poverty gap}\\  \midrule											
Average	&	0.0049	&	0.0056	&	0.0050	&	0.0023	&	0.0032	\\
10\%	&	-0.0039	&	-0.0137	&	-0.0036	&	-0.0044	&	-0.0080	\\
50\%	&	0.0022	&	-0.0031	&	0.0024	&	0.0000	&	-0.0019	\\
90\%	&	0.0107	&	0.0014	&	0.0110	&	0.0035	&	0.0005	\\
\midrule\multicolumn{6}{l}{Poverty severity}\\  \midrule											
Average	&	0.0041	&	0.0042	&	0.0042	&	0.0020	&	0.0027	\\
10\%	&	-0.0029	&	-0.0103	&	-0.0028	&	-0.0039	&	-0.0063	\\
50\%	&	0.0019	&	-0.0023	&	0.0020	&	0.0002	&	-0.0018	\\
90\%	&	0.0086	&	0.0000	&	0.0090	&	0.0030	&	-0.0001	\\
\midrule\multicolumn{6}{l}{Watts index}\\  \midrule											
Average	&	0.0118	&	0.0112	&	0.0122	&	0.0057	&	0.0074	\\
10\%	&	-0.0082	&	-0.0253	&	-0.0074	&	-0.0107	&	-0.0172	\\
50\%	&	0.0056	&	-0.0071	&	0.0061	&	0.0011	&	-0.0051	\\
90\%	&	0.0253	&	-0.0002	&	0.0265	&	0.0076	&	-0.0005	\\
\bottomrule
\end{tabular}
 \end{center}
\begin{tablenotes}
\item \footnotesize{These figures are computed across 1,146 PIP datasets. The corresponding poverty estimates are derived via numerical integration.}
\item \footnotesize{Source: authors' compilation.}
\end{tablenotes}
\end{table}

\begin{table}[tbp]
\caption{\label{ineq_app1}
Average absolute error in the estimation of the MLD and the Gini index. LIS datasets}
\centering
\begin{tabular}{lccccc}
\toprule
 & Ortega & Kakwani & SCS & $L_3$ & GQ \\
\midrule\multicolumn{6}{l}{MLD}\\  \midrule						Average	&	0.0099	&	0.0197	&	0.0093	&	0.0076	&	0.0107	\\
10\%	&	-0.0106	&	-0.0327	&	-0.0083	&	-0.0149	&	-0.0190	\\
25\%	&	-0.0045	&	-0.0255	&	-0.0037	&	-0.0091	&	-0.0134	\\
50\%	&	0.0012	&	-0.0179	&	0.0014	&	-0.0047	&	-0.0091	\\
75\%	&	0.0069	&	-0.0118	&	0.0069	&	-0.0021	&	-0.0063	\\
90\%	&	0.0197	&	-0.0078	&	0.0196	&	0.0020	&	-0.0040	\\
\midrule\multicolumn{6}{l}{Gini index}\\  \midrule											
Average	&	0.0013	&	0.0026	&	0.0011	&	0.0008	&	0.0009	\\
10\%	&	-0.0009	&	-0.0040	&	-0.0005	&	-0.0014	&	-0.0013	\\
25\%	&	-0.0002	&	-0.0033	&	-0.0001	&	-0.0010	&	-0.0010	\\
50\%	&	0.0006	&	-0.0025	&	0.0005	&	-0.0008	&	-0.0007	\\
75\%	&	0.0014	&	-0.0014	&	0.0013	&	-0.0005	&	-0.0003	\\
90\%	&	0.0029	&	0.0010	&	0.0025	&	-0.0002	&	0.0008	\\
\bottomrule
\end{tabular}
\begin{tablenotes}
\footnotesize
\item \footnotesize{These figures are computed across 589 datasets from the LIS database. MLD estimates are derived using numerical integration. Gini indices are computed using Eqs. (\ref{Gini_L0}, \ref{Gini_L1}. \ref{Gini_L2}, \ref{Gini_L3}) for the models given in (\ref{SpecialCase1}, \ref{lco}, \ref{lc2}, \ref{lc3}) respectively. For the GQ Lorenz curve, the Gini index is computed by numerical integration.}
\item Source: Authors' calculations.
\end{tablenotes}
\end{table}

\begin{table}[tbp]
\caption{\label{ineq_app2}
Average absolute error in the estimation of the MLD and the Gini index. PIP datasets}
\centering
\begin{tabular}{lccccc}
\toprule
 & Ortega & Kakwani & SCS & $L_3$ & GQ \\
\midrule\multicolumn{6}{l}{MLD}\\  \midrule						Average	&	0.0156	&	0.0137	&	0.0155	&	0.0058	&	0.0084	\\
10\%	&	-0.0099	&	-0.0315	&	-0.0093	&	-0.0152	&	-0.0190	\\
25\%	&	-0.0004	&	-0.0189	&	-0.0001	&	-0.0063	&	-0.0111	\\
50\%	&	0.0104	&	-0.0092	&	0.0104	&	-0.0010	&	-0.0058	\\
75\%	&	0.0199	&	-0.0027	&	0.0199	&	0.0011	&	-0.0035	\\
90\%	&	0.0303	&	0.0015	&	0.0300	&	0.0044	&	-0.0017	\\
\midrule\multicolumn{6}{l}{Gini index}\\  \midrule											
Average	&	0.0021	&	0.0017	&	0.0019	&	0.0007	&	0.0006	\\
10\%	&	-0.0003	&	-0.0034	&	-0.0003	&	-0.0013	&	-0.0009	\\
25\%	&	0.0006	&	-0.0022	&	0.0006	&	-0.0009	&	-0.0007	\\
50\%	&	0.0019	&	-0.0008	&	0.0017	&	-0.0007	&	-0.0003	\\
75\%	&	0.0029	&	0.0004	&	0.0028	&	-0.0004	&	0.0002	\\
90\%	&	0.0040	&	0.0018	&	0.0037	&	-0.0001	&	0.0007	\\
\bottomrule
\end{tabular}
\begin{tablenotes}
\footnotesize
\item \footnotesize{These figures are computed across 1,146 datasets from the PIP database. MLD estimates are derived using numerical integration. Gini indices are computed using Eqs. (\ref{Gini_L0}, \ref{Gini_L1}. \ref{Gini_L2}, \ref{Gini_L3}) for the models given in (\ref{SpecialCase1}, \ref{lco}, \ref{lc2}, \ref{lc3}) respectively. For the GQ Lorenz curve, the Gini index is computed by numerical integration.}
\item Source: Authors' calculations.
\end{tablenotes}
\end{table}

\begin{figure}[tbhp]
\centering
\begin{tabular}{cc}
\includegraphics[scale=0.75]{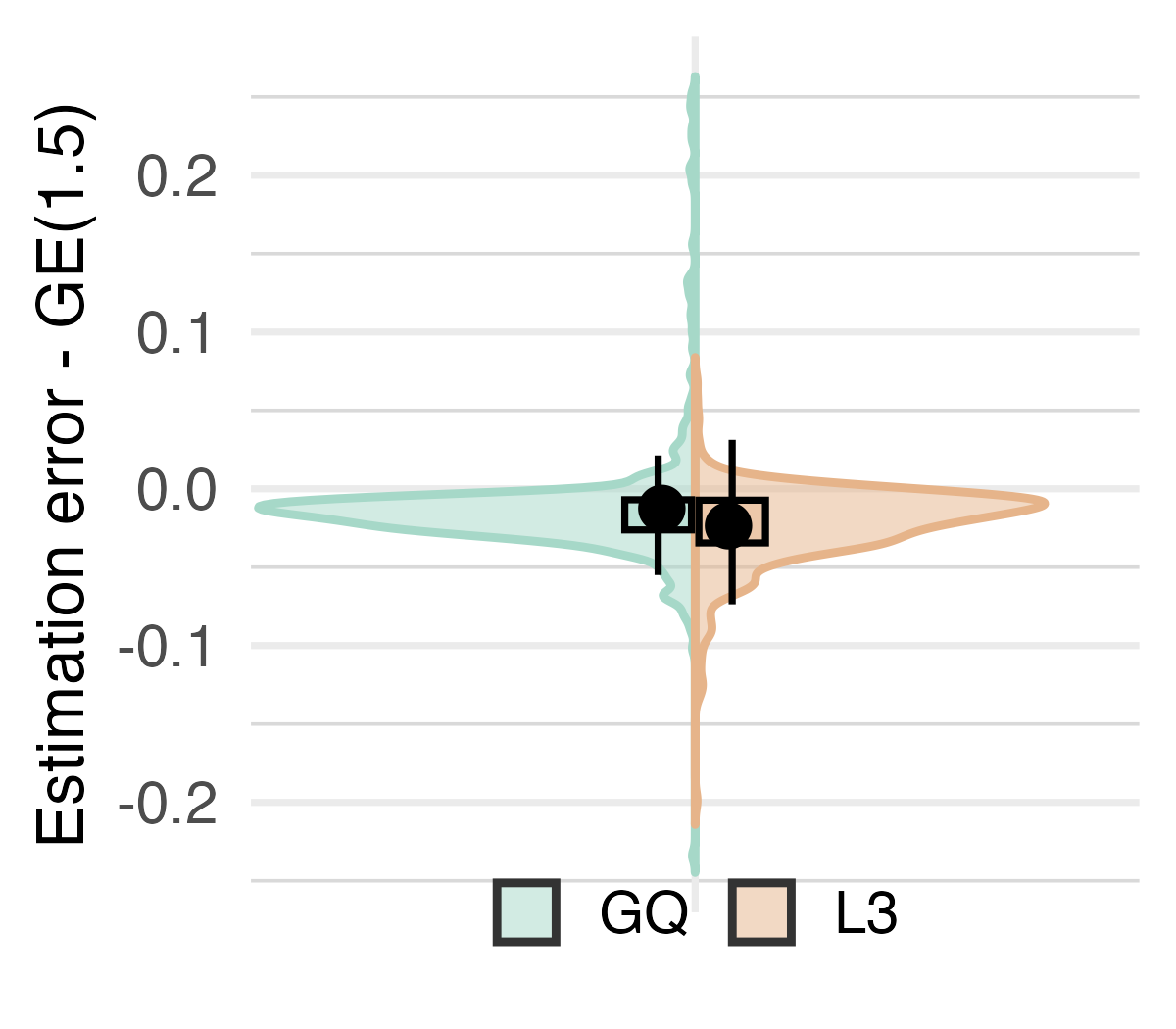} &
\includegraphics[scale=0.75]{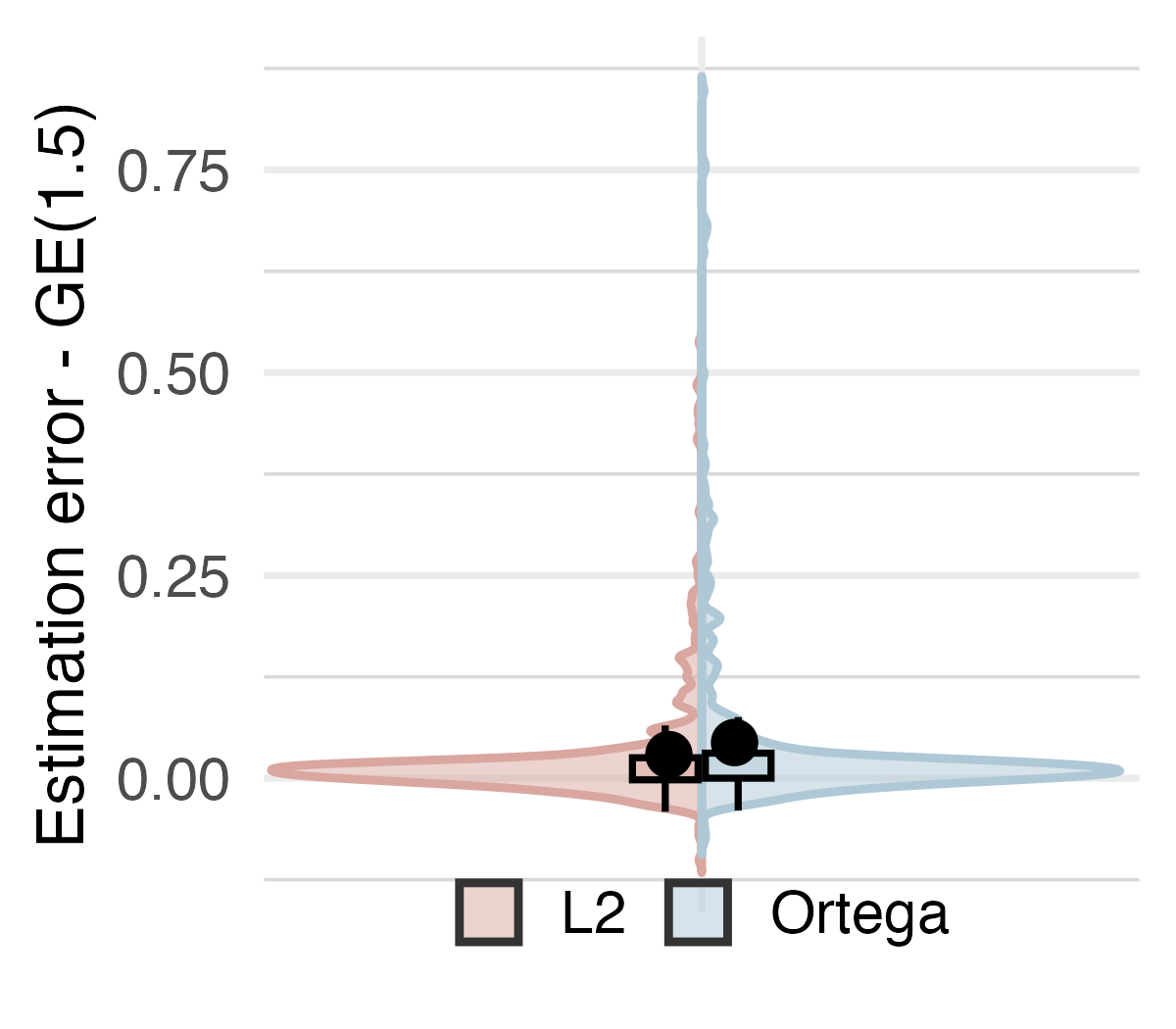}
\end{tabular}
\caption{Distribution of the estimation error for the GE(1.5) measure}
\label{GE15_v}
\end{figure}

\begin{table}[tb]
\caption{\label{bias_models_b}
Average absolute bias in the estimation of poverty and inequality measures using alternative Lorenz curve specifications: LIS and PIP data}
 \begin{center}
\begin{tabular}{l c c c c c}
\toprule
	&Ortega	&Kakwani 	     &SCS	&$L_3$	&GQ\\
    	&       	&($\alpha = 1$)	    &  	&       	&       \\								
\midrule\multicolumn{6}{l}{MLD}\\  \midrule											
$N = 500$	&	0.00094	&	0.00102	&	0.00405	&	0.00192	&	0.00432	\\
$N = 2500$	&	0.00041	&	0.00043	&	0.00205	&	0.00067	&	0.00135	\\
$N = 5000$	&	0.00029	&	0.00032	&	0.00145	&	0.00045	&	0.00083	\\
\midrule\multicolumn{6}{l}{Gini index}\\  \midrule											
$N = 500$	&	0.00164	&	0.00161	&	0.00407	&	0.00166	&	0.00322	\\
$N = 2500$	&	0.00045	&	0.00046	&	0.00193	&	0.00050	&	0.00086	\\
$N = 5000$	&	0.00028	&	0.00029	&	0.00134	&	0.00032	&	0.00051	\\
\midrule\multicolumn{6}{l}{SPR}\\  \midrule											
$N = 500$	&	0.00050	&	0.00082	&	0.00110	&	0.00070	&	0.00060	\\
$N = 2500$	&	0.00019	&	0.00024	&	0.00054	&	0.00027	&	0.00022	\\
$N = 5000$	&	0.00013	&	0.00016	&	0.00041	&	0.00018	&	0.00014	\\
\bottomrule
\end{tabular}
 \end{center}
\begin{tablenotes}
\item \footnotesize{These figures are computed across 1,735 datasets, comprising 1,146 from the PIP database and 589 from the LIS database. Biases are estimated by Monte Carlo simulation of $10^3$ samples of size $N = 500, 2500, 5000.$}
\item \footnotesize{Source: authors' compilation.}
\end{tablenotes}
\end{table}

\begin{table}[tbp]
\caption{\label{bias_models_p}
Average absolute bias in the estimation of poverty measures using alternative Lorenz curve specifications. PIP data}
 \begin{center}
\begin{tabular}{l c c c c c}
\toprule
	&Ortega	&Kakwani 	     &SCS	&$L_3$	&GQ\\
    	&       	&($\alpha = 1$)	    &  	&       	&       \\								
\midrule\multicolumn{6}{l}{Poverty headcount}\\  \midrule											
$N = 500$	&	0.00069	&	0.00077	&	0.00194	&	0.00190	&	0.00217	\\
$N = 2500$	&	0.00027	&	0.00026	&	0.00093	&	0.00066	&	0.00073	\\
$N = 5000$	&	0.00018	&	0.00018	&	0.00065	&	0.00043	&	0.00047	\\
\midrule\multicolumn{6}{l}{Poverty gap}\\  \midrule											
$N = 500$	&	0.00061	&	0.00087	&	0.00111	&	0.00243	&	0.00223	\\
$N = 2500$	&	0.00027	&	0.00034	&	0.00081	&	0.00077	&	0.00074	\\
$N = 5000$	&	0.00020	&	0.00022	&	0.00073	&	0.00049	&	0.00048	\\
\midrule\multicolumn{6}{l}{Poverty severity}\\  \midrule											
$N = 500$	&	0.00283	&	0.00439	&	0.00585	&	0.00352	&	0.00658	\\
$N = 2500$	&	0.00288	&	0.00335	&	0.00238	&	0.00161	&	0.00277	\\
$N = 5000$	&	0.00280	&	0.00321	&	0.00153	&	0.00114	&	0.00188	\\
\midrule\multicolumn{6}{l}{Watts index}\\  \midrule											
$N = 500$	&	0.06544	&	0.06364	&	0.01284	&	0.00816	&	0.01523	\\
$N = 2500$	&	0.05408	&	0.05415	&	0.01117	&	0.00358	&	0.00762	\\
$N = 5000$	&	0.05136	&	0.05249	&	0.01225	&	0.00257	&	0.00544	\\
\bottomrule
\end{tabular}
 \end{center}
\begin{tablenotes}
\item \footnotesize{These figures are computed across 1,146 datasets from the PIP database. Biases are estimated by Monte Carlo simulation of $10^3$ samples of size $N = 500, 2500, 5000.$}
\item \footnotesize{Source: authors' compilation.}
\end{tablenotes}
\end{table}

\begin{table}[tbp]
\caption{\label{bias_models_l}
Average absolute bias in the estimation of inequality measures using alternative Lorenz curve specifications. LIS data}
 \begin{center}
\begin{tabular}{l c c c c c}
\toprule
	&Ortega	&Kakwani 	     &SCS	&$L_3$	&GQ\\
    	&       	&($\alpha = 1$)	    &  	&       	&       \\								
\midrule\multicolumn{6}{l}{Atkinson index (1)}\\  \midrule											
$N = 500$	&	0.00073	&	0.00084	&	0.00208	&	0.00150	&	0.00231	\\
$N = 2500$	&	0.00029	&	0.00027	&	0.00105	&	0.00046	&	0.00084	\\
$N = 5000$	&	0.00020	&	0.00018	&	0.00073	&	0.00028	&	0.00055	\\
\midrule\multicolumn{6}{l}{Atkinson index (1.5)}\\  \midrule											
$N = 500$	&	0.00063	&	0.00093	&	0.00118	&	0.00195	&	0.00229	\\
$N = 2500$	&	0.00026	&	0.00037	&	0.00084	&	0.00057	&	0.00082	\\
$N = 5000$	&	0.00020	&	0.00026	&	0.00074	&	0.00035	&	0.00053	\\
\midrule\multicolumn{6}{l}{Theil index}\\  \midrule											
$N = 500$	&	0.00176	&	0.00224	&	0.00803	&	0.00316	&	0.00759	\\
$N = 2500$	&	0.00075	&	0.00080	&	0.00442	&	0.00121	&	0.00351	\\
$N = 5000$	&	0.00061	&	0.00066	&	0.00318	&	0.00079	&	0.00247	\\
\midrule\multicolumn{6}{l}{GE(1.5)}\\  \midrule											
$N = 500$	&	0.04047	&	0.04738	&	0.01633	&	0.00847	&	0.01877	\\
$N = 2500$	&	0.01620	&	0.02220	&	0.00968	&	0.00299	&	0.01036	\\
$N = 5000$	&	0.01181	&	0.01557	&	0.00795	&	0.00190	&	0.00774	\\
\bottomrule
\end{tabular}
 \end{center}
\begin{tablenotes}
\item \footnotesize{These figures are computed across 589 datasets from the LIS database. Biases are estimated by Monte Carlo simulation of $10^3$ samples of size $N = 500, 2500, 5000.$}
\item \footnotesize{Source: authors' compilation.}
\end{tablenotes}
\end{table}

\end{document}